\documentclass[11pt,onecolumn,draftclsnofoot]{IEEEtran}
\usepackage[utf8]{inputenc}
\usepackage[T1]{fontenc}
\usepackage{url}
\usepackage{ifthen}
\usepackage{cite}
\usepackage{bbm}
\usepackage[cmex10]{amsmath} 

\usepackage{graphicx}
\usepackage{amsfonts}
\usepackage{algorithmicx}
\usepackage{algpseudocode}
\usepackage{algorithm}
\usepackage{color}
\usepackage{verbatim}
\usepackage{graphicx,wrapfig,lipsum}
\usepackage[colorinlistoftodos]{todonotes}
\usepackage{amssymb}
\usepackage{amsthm}
\usepackage{comment}
\usepackage{mathtools}
\usepackage{bm}
\usepackage{version}
\usepackage{dsfont}
\graphicspath{{fig/}}

\usepackage{url}
\excludeversion{short}
\includeversion{draft}
\includeversion{short2}

\newcommand{\DG}[1]{#1}

\newcommand{\E}[2]{E^{\delta}_{#1,#2}}
\newcommand{\Eth}[2]{E^{\frac{19\delta}{20}}_{#1,#2}}
\newcommand{\Jth}[2]{J^{\frac{19\delta}{20}}_{#1,#2}}
\newcommand{\J}[2]{J^{\delta}_{#1,#2}}
\newcommand{\G}[2]{G^{\delta}_{#1,#2}}
\newcommand{\Ea}[2]{A^{\delta}_{#1,#2}}
\newcommand{\Eb}[2]{B^{\delta}_{#1,#2}}
\newcommand{\Ec}[2]{C^{\delta}_{#1,#2}}
\newcommand{\Ed}[2]{D^{\delta}_{#1,#2}}
\newcommand{\Gj}[2]{G^{\delta,j}_{#1, #2}}

\newcommand{\coef}[0]{\frac{19\delta}{20}}
\newcommand{\coefff}[0]{\frac{\delta}{40}}
\newcommand{\coeff}[0]{\frac{80}{\delta}}
\newcommand{\coeft}[0]{\frac{80\Delta}{\delta}}
\newcommand{\coefk}[0]{\frac{160\alpha(1+\Delta)}{\delta}}

\newcommand{\growth}[0]{\left(1-\frac{41}{40}\delta\right)}

\newcommand{\N}[2]{N_{#1,#2}}
\newcommand{\X}[2]{X_{#1,#2}}
\newcommand{\Y}[2]{Y_{#1,#2}}
\newcommand{\Z}[2]{Z_{#1,#2}}
\newcommand{\Ej}[2]{E^{\delta,j}_{#1,#2}}
\newcommand{\Ezero}[2]{E^{\delta,0}_{#1,#2}}
\newcommand{\Gzero}[2]{G^{\delta,0}_{#1,#2}}

\newcommand{\Eaj}[2]{A^{\delta,j}_{#1,#2}}
\newcommand{\Ebj}[2]{B^{\delta,j}_{#1,#2}}
\newcommand{\Ecj}[2]{C^{\delta,j}_{#1,#2}}
\newcommand{\Edj}[2]{D^{\delta,j}_{#1,#2}}
\newcommand{\Mj}[1]{M^j_{#1}}
\newcommand{\Nj}[2]{N^j_{#1,#2}}
\newcommand{\Xj}[2]{X^j_{#1,#2}}
\newcommand{\Yj}[2]{Y^j_{#1,#2}}
\newcommand{\Zj}[2]{Z^j_{#1,#2}}
\newcommand{\Xzero}[2]{X^0_{#1,#2}}

\newcommand{\Zzero}[2]{Z^0_{#1,#2}}
\newcommand{\Nzero}[2]{N^0_{#1,#2}}
\newcommand{\f}[1]{f_{#1}}
\newcommand{\fj}[1]{f^j_{#1}}
\newcommand{\h}[1]{h(#1)}
\newcommand{\hj}[1]{h^j(#1)}
\newcommand{\bl}[1]{b_{#1}}
\newcommand{\U}[1]{{X_{#1}}}
\newcommand{\V}[1]{{Y_{#1}}}

\newcommand{\ceils}[0]{\lceil s \rceil}
\newcommand{\floort}[0]{\lfloor t \rfloor}

\newcommand{\credible}[0]{credible }



\newtheorem{theorem}{Theorem}
\newtheorem{lemma}[theorem]{Lemma}
\newtheorem{proposition}[theorem]{Proposition}
\newtheorem{corollary}[theorem]{Corollary} 
\newtheorem{definition}[theorem]{Definition}

\begin{document}

\title{Continuous-Time Analysis of the Bitcoin and Prism Backbone Protocols}

\author{\IEEEauthorblockN{Jing Li\textsuperscript{1} and Dongning Guo\textsuperscript{2}} \\
\IEEEauthorblockA{{Department of Electrical and Computer Engineering} \\
{Northwestern University}\\
Evanston, IL 60208 \\
\textsuperscript{1}jingli2015@u.northwestern.edu, \\ \textsuperscript{2}dGuo@northwestern.edu\\}
}

\maketitle

\begin{abstract}
Bitcoin  is  a  peer-to-peer  payment  system  proposed by   Nakamoto   in   2008.   Based   on   the   Nakamoto   consensus,  Bagaria,  Kannan,  Tse,  Fanti, and  Viswanath proposed the Prism protocol in  2018 and showed that it  achieves near-optimal blockchain throughput while   maintaining \DG{a}  similar   level  of  security as bitcoin. Previous probabilistic security guarantees for the bitcoin and Prism backbone protocols were either established under a simplified discrete-time model or expressed in  terms of exponential order result\DG{s}.
This paper presents  a  streamlined and strengthened analysis under a more realistic continuous-time model.
A fully rigorous model for blockchains is developed with no \DG{restrictions} on adversarial miners except for an upper bound on their aggregate mining rate.
\DG{The only assumption on the peer-to-peer network is that all block propagation delays are upper bounded by a constant.}
\DG{A new notion of ``$t$-credible blockchains'' is introduced, which, together with some carefully defined} ``typical'' events concerning block production
over time interval\DG{s, is crucial to establish probabilisitic security guarantees in continuous time}.  A blockchain growth theorem, a blockchain quality theorem, and a common prefix theorem are established with explicit probability bounds.
\DG{Moreover,} under a certain typical event which occurs with probability close to $1$, a \DG{valid} transaction that is deep enough in one credible blockchain 
is shown to be permanent in the sense that it must be found  in all future credible blockchains.
\end{abstract}

\section{Introduction}
\subsection{The bitcoin backbone protocol}
Bitcoin was invented by Nakamoto\cite{nakamoto2008bitcoin} in 2008 as an electronic payment system. The system is built on a distributed ledger technology commonly referred to as blockchain.
A blockchain is a finite sequence of transaction-recording blocks which begins with a genesis block, and every subsequent block contains a cryptographic hashing of the previous one (which confirms all preceding blocks). To mine a block requires proof of work: A nonce must be included such that the block's
hash value satisfies a difficulty requirement.
Miners join a peer-to-peer network to inform each other of new blocks.
An honest miner follows the longest-chain rule, i.e., it always tries to mine a block at the maximum height.

Different blocks may be mined and announced at around the same time. So honest miners may extend different blockchains depending on which blocks they hear first. This phenomenon is called forking, which must be resolved quickly to reach timely consensus about the ledger.

An adversarial miner may wish to sabotage consensus or manipulate the network to a consensus to its own advantage. In particular, forking presents opportunities for double spending, which is only possible if a transaction included in the longest fork at one time is not included in a different fork that overtakes the first one to become the longest blockchain.
Nakamoto\cite{nakamoto2008bitcoin} characterized the race between the honest miners and the adversary as a random walk with a drift. Nakamoto showed that the probability the adversary blockchain overtakes the honest miner's consensus blockchain vanishes exponentially over time as long as the collective mining power of adversarial miners is less than that of honest miners.
In this case, a bitcoin transaction becomes (arbitrarily) secure if it is confirmed by enough new blocks.


Garay, Kiayias, and Leonardos\cite{garay2015bitcoin} first formally described and analyzed the bitcoin backbone protocol under the lockstep synchronous model, where all miners have perfectly synchronized rounds and all miners receive the same block(s) at exactly the end of the round. Under this model, \cite{garay2015bitcoin} established a blockchain quality theorem, which states the honest miners contribute at least a certain percentage of the blocks with wish probability.
Also established in \cite{garay2015bitcoin} is a common prefix theorem, which states if a block is $k$ blocks deep in an honest miner's blockchain, then the block is in  all other honest miners' blockchains with high probability (the probability that some honest miner does not extend this block vanishes exponentially with $k$). Kiayias and Panagiotakos\cite{kiayias2015speed} established a blockchain growth theorem, which quantifies the number of blocks added to the blockchain during any time interval.  The  blockchain  growth theorem and the  blockchain  quality theorem guarantee that many honest blocks will eventually become $k$ deep in an honest miner's blockchain (liveness). The common prefix theorem then guarantees that an honest miner's $k$-deep block become permanent consensus of all honest miners (consistency). Thus, every transaction that is recorded in a sufficiently deep block in an honest miner's blockchain is with high probability guaranteed to remain in the {transaction} ledger.

The strictly lockstep synchrony model completely assumes away  network delay and failure.
Several meaningful analyses have been proposed under the non-lockstep synchrony model, where messages can be delayed arbitrarily but the delay is upper bounded. A complicated analysis with strong assumptions \cite{pass2017analysis} showed that the blockchain growth theorem, the blockchain  quality theorem, and the common prefix theorem remain valid under the non-lockstep synchrony model. Reference \cite{kiffer2018better} also reasoned the consistency of bitcoin protocol using the Markov chains, although their result has a non-closed form.
Most previous analyses \cite{garay2015bitcoin,garay2017bitcoin, pass2017analysis, bagaria2019prism} assume the  blockchain's lifespan is finite, i.e., there exists a maximum round when the blockchain ends. In \cite{li2019onanalysisarxiv}, we dropped the finite horizon assumption and proved stronger properties of the bitcoin backbone protocol regardless of whether or not the blockchains have a finite lifespan.


Most previous work\cite{garay2015bitcoin, garay2017bitcoin,pass2017analysis, bagaria2019prism, Ren2019Analysis, niu2019analysis} expressed the probability of the mentioned properties in exponential order result (\DG{using} big \DG{$O$ or big $\Omega$} notation). Our previous work \cite{li2019onanalysisarxiv} gives the explicit bounds for the liveness and consistency under the non-lockstep synchronous discrete-time model. The strategies taken by previous works can be described as the following:
Intuitively, during any time interval the liveness and consistency of honest blockchains hold under the following conditions: 1) The number of honest blocks mined during this time interval is larger than the number  of adversarial blocks mined, so that the longest blockchain will not be overtaken by the adversarial party. 2) There are enough number of non-reversible honest blocks to guard the honest blockchain, in case the adversarial use strategies (like selfish mining) to introduce disagreement between honest miners and split their hashing power. Technically, with respect to time interval $[s,t]$, a ``good event'' occurs if the numbers of various blocks mined during the period are close to  their respective expected values. A ``typical event'' with respect to $[s,t]$ occurs if good events occur for all time intervals covering $[s,t]$, so that the consistency of honest blockchains is guaranteed from time $t$ onward.
The desired properties hold  under the typical events, which are shown to almost certainly occur in the discrete-time model.

The  discrete-time  model  eases  analysis  but is  still  a significant  departure  from  reality.
In 2019, Ren \cite{Ren2019Analysis} extended the liveness and consistency of bitcoin protocol assuming the continuous-time model where mining is modelled as a Poisson point process. The probability bounds are shown to be exponential in a linear order term in the confirmation time.

In this paper, we \DG{build a simple stochastic model for continuous-time block mining processes and the resulting blockchains.  We impose no restrictions on adversarial miners except for an upper bound on their aggregate mining rate.}
In addition, \DG{the only assumption on the peer-to-peer network is that all block propagation delays are upper bounded by a constant.}
We introduce the \DG{new notion of a} {\em $t$-credible blockchain} to describe 
a blockchain that an honest miner may adopt \DG{at time $t$}.
We also develop a technique to \DG{analyze the probability of} 
the intersection of uncountably many good events with 
\DG{continuous} starting and ending points. 
\DG{Using a sequence of lemmas, we} derive explicit \DG{bounds as probabilistic guarantees of} the liveness and consistency \DG{of the bitcoin backbone protocol.  These results} are more refined than previous exponential order results.

\DG{We note that} several existing proofs in the literature (including some of our own earlier work) are flawed.  \DG{A recurrent subtle mistake is to presume memorylessness of the mining process over a time interval defined according to some miners' views and actions.  The boundaries of such an interval are in fact very complicated (random) stopping times.  As an extreme but illuminating example, we argue that the average mining rate between the time of genesis and the time of the first blockchain forking (this is a stopping time) is expected to be lower than the long-term average, because as soon as two blocks are mined close to each other in time, the blockchain is likely to fork.  In general, it is perilous to work with miners' views, which depend on network topology as well as the adversarial strategy.  In contrast, the framework developed here allows us to} prove all results rigorously without \DG{explicitly defining honest and adversarial miners' views and actions.  Instead, all consequential views and actions are reflected in the well-defined published and unpublished blockchains.}

\begin{short}
Due to space limitations, the proof of some results are relegated to a longer version of this paper\cite{li2020probablistic}.
\end{short}

\subsection{The Prism protocol}
\DG{It is well known that} the throughput of bitcoin is \DG{severely restricted} by design to ensure security\cite{decker2013information}.
In particular, the average time interval between new blocks is set to be much longer than the block propagation delays so that forking is infrequent\cite{sompolinsky2015secure}.
Many ideas have been proposed to improve the blockchain throughput.
One way is to construct high-forking blockchains by
optimizing the forking rule, which is vulnerable to certain attacks\cite{sompolinsky2015secure,lewenberg2015inclusive, sompolinsky2016spectre, sompolinsky2018phantom,natoli2016balance, li2018scaling,zheng2018blockchain}. Another line of work is to decouple the various functionalities of the blockchain\cite{eyal2016bitcoin, pass2017fruitchains}, under the spirit of which
Bagaria, Kannan, Tse, Fanti, and Viswanath\cite{bagaria2019prism} proposed the Prism protocol in $2018$. The Prism protocol defines one proposer blockchain and many voter blockchains. The voter blocks elect a leader block at each level of the proposer blockchain by vote. The sequence of leader blocks concludes the contents of all voter blocks, and finalizes the ledger.
A voter blockchain follows the bitcoin protocol to provide security to leader election process.
With this design, the throughput (containing the content of {\em all} voter blocks) is decoupled from the mining rate of each voter blockchain. Slow mining rate guarantees the security of each voter blockchain as well as the leader sequence they selected.
Prism achieves security against up to 50\% adversarial hashing power, optimal throughput up to the capacity of the network, and fast confirmation latency for  transactions. A thorough description and analysis is found in \cite{bagaria2019prism}.

In \cite{bagaria2019prism}, liveness and consistency of  Prism transactions were proved assuming a finite life span of the blockchains under the lockstep synchrony model\cite{bagaria2019prism}. In
\cite{li2019onanalysisarxiv} we have  strengthened and extended the results to the  non-lockstep synchrony model.
\DG{In this paper, we establish} the key properties \DG{of the Prism backbone protocol under the more realistic} continuous-time model.

\section{The Bitcoin Backbone Protocol}

\DG{In order to develop a fully rigorous analysis, we first build an explicit model for the blockchain system which evolves in time according to the bitcoin backbone protocol.  Some existing models in the literature involve miners' views, their protocol executions, block propagation, adversarial miners' control power, etc.
Unfortunately, the adversary's strategy space is essentially impossible to exhaust.  In particular, the adversarial miners may regulate their mining rates and forks according to the the honest mining outcome.  Consequently, it is very hard to precisely describe the joint distribution of the honest and the adversarial mining processes.  Some authors make the unrealistic assumption that the adversarial mining processes are homogeneous in time, which of course severely weakens their security guarantees.

In this paper, we think of the adversarial strategy as a policy which maps the mining history up to any point in time to adversarial actions, which control adversarial block arrivals, their parents, and their publication times. Without loss of generality, we also regard the (possibly dynamic) topology of the peer-to-peer mining network and all network uncertainties as components of the adversarial strategy.

Once the adversarial strategy $v$ is fixed, the remaining uncertainties in the entire blockchain system are completely described by a probability space $(\Omega_v, \mathcal{F}_v, P_v)$.  Each $\omega \in \Omega_v$ represents the outcome of the entire system from time $0$ to eternity, including all block arrival times, their parents, and their publication times.  Every block can be traced back to the genesis block by recursive parental reference.  This finite sequence of blocks along the path to the genesis block is referred to as a blockchain.  It is important to note that the probability space does not include miners' views and actions.
Rather, consequential views and actions are reflected in the published and unpublished blockchains over time.

In this model, the aggregate honest mining process is a homogeneous Poisson point process.  The adversarial mining process is arbitrary aside from the restriction that the arrivals of adversarial blocks are ``dominated'' by a homogeneous Poisson point process in some probabilistic sense.  In addition, the (adversarial) network uncertainties lead to arbitrary random block propagation delays, which are upper bounded by a constant.
As we shall see, such this simple model is sufficient as far as liveness and consistency properties of blockchains are concerned.
}


{\subsection{Model}}
\DG{In this subsection, we fill in the details of the probability space $(\Omega_v,\mathcal{F}_v,P_v)$.}

\DG{Throughout this paper,} by saying “by time $t$” we mean all time up to and including time $t$ starting from time $0$ ({excluding} time $0$).
{Hence ``by time $t$'' is equivalent to ``during $(0, t]$''.}

\begin{definition}
{(Block\footnote{
\DG{A block in a practical blockchain system is a data structure with an identifier and a reference to its parent block.  As long as the identifier consists of a large number of bits, it is fair to assume that each block has a unique identifier for all practical purposes.  Note, however, it is not possible for the data structure to include its precise mining time or its (universal) block number.}
} and mining process)}
An honest genesis block, also referred to as block $0$, is mined at time $0$. Subsequent blocks are referred to as block $1$, block $2$, and so on, in the order they are mined in time after time $0$.
For $t>0$, let $M_t$ denote the
total number of non-genesis blocks mined {\color{black}by} time $t$.
If a single block is mined at time $t$, it must be block $M_t$. If $k>1$ blocks are mined at exactly the same time $t$, we assume the tie is broken in some deterministic manner so that their block numbers are $M_{t}-k+1, \ldots, M_t$, respectively.\footnote{We address ties for mathematical rigor. Ties essentially do not happen in continuous time.}
\end{definition}

\begin{definition}
(Blockchain)
We use $\f{k} \in \{0,1,...,k-1\}$ to denote block $k$'s parent block number.
For block $k$ to be valid, there must exist a unique sequence of block numbers $\bl{0},\bl{1}, \dots, \bl{n}$ where $\bl{0}=0$, $\bl{n}=k$, and $\f{\bl{i}}=\bl{i-1}$ for $i =1,\ldots ,n$.  This sequence is referred to as blockchain $(\bl{0},\ldots,\bl{n})$ or simply blockchain $k$ since it is determined by $k$.
\end{definition}
We assume that block $k$ is {validated by} the existence of the entire blockchain $k${:}
To validate a block $k$, one needs access to the entire blockchain $k$. Because invalid blocks are inconsequential as far as the distributed consensus protocol is concerned, throughout this paper, by a block we always mean a valid block unless noted otherwise.

\begin{definition}
{(Height)} The height of block $k$, denoted as $\h{k}$, is defined as the height of blockchain $k$, which is in turn defined as the number of non-genesis blocks in it.
\end{definition}

We let $T_k$  denote the time when block $k$ is mined.  We say blockchain $k$ is mined by time $t$ if $T_k\le t$.  We let $P_k$ denote the time when block $k$ is published.
A blockchain is said to be mined by time $t$ if all of its blocks are mined by time $t$.
A blockchain is said to be published by time $t$ if all of its blocks are published by time $t$.
Let $\Delta$ denote an upper bound for all communication delays.

\begin{definition} \label{def: c longest blockchain}
($t$-credible blockchain) Blockchain $b$ is said to be \emph{$t$-credible} {if the blockchain} has been published by time $t$, and is no shorter than any blockchain published by time $t-\Delta$. That is to say,
{
\begin{align}
P_{b} &  \le t,
\end{align}
and
\begin{align}
\h{b} & \ge \h{k}, \quad \forall k: P_{k} \le t - \Delta. \label{equ: c height of longest blockchain}
\end{align}
}
If there is no need to specify $t$ explicitly, blockchain $b$ can also be simply called a credible blockchain.
\end{definition}

There can be multiple $t$-credible blockchains, which may or may not be of the same {height}.

According to the bitcoin protocol, if block $k$ is honest, it must extend a $T_k$-credible blockchain and it must be published as soon as it is mined, i.e., $P_k = T_k$.
\begin{lemma}\label{lemma: honest block append to honest blockchain}
If block $k$ is honest, then both blockchain $\f{k}$ and blockchain $k$ must be $T_k$-credible.
\end{lemma}
\begin{proof}
Since an honest block extends a $T_k$-credible blockchain,  blockchain $\f{k}$ is $T_k$-credible. The height of blockchain $k$ is greater than the height of every blockchain published by {time} $T_k-\Delta$, thus blockchain $k$ is also $T_k$-credible.
\end{proof}

Although an honest block always extends a credible blockchain, a credible blockchain {may not} end with an honest block.
An adversarial block may or may not extend a credible blockchain and may be published any time after it is mined.

Let $N_t$ denote the total number of honest blocks mined {during $(0, t]$}. We assume the sum mining rate of honest miners is $\alpha$, then $(N_{t}, t\ge 0)$ is a homogeneous Poisson point process with rate $\alpha$.
Let $Z_t$ denote the total number of adversarial blocks mined {during $(0, t]$}, then $N_t + Z_t = M_t$.
We assume the sum mining rate of all adversaries is no larger than $\beta$, so
{\color{black}the number of adversarial blocks mined during any time interval is upper bounded
probabilistically. Specifically,} the probability that the number of adversarial blocks mined during $(s,t]$ is greater than {a number} is upper bounded by the probability that a Poisson distribution with parameter $\beta(t-s)$ is greater than {the same number}. That is to say,  for real number $a$ and $0\le s<t$,
\begin{align}
    P(Z_{t}-Z_s\le a) \ge e^{-\beta(t-s)} \sum_{i=0}^{\lfloor a \rfloor}\frac{(\beta(t-s))^i}{i!}. \label{equ: c Z bound}
\end{align}

We note that an overarching probability space can be defined for arbitrary given adversary strategies.
The key to a simple analysis is to use property \eqref{equ: c Z bound}, which holds regardless of the adversarial strategies, the network topology, and other sources of randomness (e.g., communication loss and latency).

\begin{definition}\label{def: lagger}
{(Lagger and loner)} An honest block {$k$} is called a {\em lagger} if it is the only honest block mined during {$[T_k-\Delta, T_k]$}.
The lagger is also called a {\em loner} if it is also the only honest block mined during $[{\color{black}T_k}, {\color{black}T_k}+\Delta]$.
\end{definition}

{Suppose $0\le s < t$.} Let $\N{s}{t}= N_{t}-N_s$ denote the total number of honest blocks mined during time interval $(s,t]$. Let $\X{s}{t}$ denote the total number of laggers mined during  $(s, t]$. Let $\Y{s}{t}$ denote the total number of loners mined during $(s, t]$. Let $\Z{s}{t}$ denote the total number of adversarial blocks mined during $(s, t]$. {By convention, $\N{s}{t} = \X{s}{t} = \Y{s}{t} = \Z{s}{t} = 0$ for all $s \ge t$.}


Define random variable $\U{i} = 1$ if the $i$-th honest block is a lagger, and $0$ otherwise.
Then we have {
\begin{align}\label{equ: XU}
\X{s}{t} = \sum_{i=N_s+1}^{N_t}\U{i}.
\end{align}
Likewise, denote $\V{i} = 1$ if the $i$-th honest block is a loner and $0$ otherwise. Then we have
\begin{align} \label{equ: YV}
\Y{s}{t} = \sum_{i=N_s+1}^{N_t}\V{i}.
\end{align}
}

For convenience, we introduce the following parameter negatively related to the maximum propagation delay and the total honest mining rate:
\begin{align}\label{def: g}
    g=e^{-\alpha \Delta}.
\end{align}

{We further make a crucial assumption that} the parameters satisfy
\begin{align}
    (1-\frac{81}{40}\delta) g^2\alpha > \beta \label{equ: alpha beta}
\end{align}
{where $\delta$ is a constant on $(0,\frac{40}{81})$}.
This assumption indicates the {sum} mining rate  of adversarial blocks must be strictly less than the mining rate of honest blocks subject to  a ``propagation discount'' {($g^2=e^{-2\alpha \Delta}$)} and also a penalty dependent on the typicality factor $\delta$.  Also, throughout this paper we assume $\alpha > \frac{1}{2}$ block per time unit. This requirement can be satisfied by adjusting the time {units} we adopt.



{\subsection{Preliminaries}}
Next, we introduce a few preliminaries.
\begin{lemma} \label{lemma: Poisson property}
Let $X$ be a Poisson random variable with parameter $\lambda$. Then for every $\delta \in (0, 1]$,
\begin{align}
    P(X \le (1-\delta)\lambda) < e^{{\color{black}-\frac{1}{2}\delta^2\lambda}}, \label{c 0.0}
\end{align}
and
\begin{align}
    P(X \ge (1+\delta)\lambda) < e^{-\frac{1}{3}\delta^2\lambda}. \label{c 0.1}
\end{align}
\end{lemma}

\begin{draft}
\begin{proof}
To prove \eqref{c 0.0}, we have
\begin{align}
  P(X \le (1-\delta)\lambda) & = P(e^{-tX} \ge e^{-t(1-\delta)\lambda}) \label{c 0.05}\\
  & \le \frac{\mathbb{E}[e^{-tX}]}{e^{-t(1-\delta)\lambda}} \label{c 0.06}\\
  & = e^{(e^{-t}-1)\lambda + t(1-\delta)\lambda} \label{c 0.07}
\end{align}
where \eqref{c 0.06} is due to Markov inequality and \eqref{c 0.07} is due to the moment generating function of Poisson random variable. Picking $t = -\log(1-\delta)$, we have
\begin{align}
      P(X \le (1-\delta)\lambda) & \le e^{-\delta - (1-\delta)\log(1-\delta)} \\
      & < e^{-\frac{1}{2}\delta^2\lambda}  \label{c 0.08}
\end{align}
where \eqref{c 0.08} is due to $(1-\delta)\log(1-\delta) > -\delta + \frac{\delta^2}{2}$ for $\delta \in (0, 1)$.

To prove \eqref{c 0.1}, we have
\begin{align}
  P(X \ge (1+\delta)\lambda) & = P(e^{tX} \ge e^{t(1+\delta)\lambda}) \label{c 0.01}\\
  & \le \frac{\mathbb{E}[e^{tX}]}{e^{t(1+\delta)\lambda}} \label{c 0.02}\\
  & = e^{(e^t-1)\lambda - t(1+\delta)\lambda} \label{c 0.03}
\end{align}
where \eqref{c 0.01} is due to Markov inequality and \eqref{c 0.02} is due to the moment generating function of Poisson random variable. Picking $t = \log(1+\delta)$, we have
\begin{align}
    P(X \ge (1+\delta)\lambda) & \le e^{\delta - (1+\delta)\log(1+\delta)} \\
    & < e^{-\frac{1}{3}\delta^2\lambda} \label{c 0.04}
\end{align}
where \eqref{c 0.04} is due to $(1+\delta)\log(1+\delta) > \delta + \frac{\delta^2}{3}$ for $\delta \in (0, 1)$.
\end{proof}
\end{draft}

\begin{proposition} \label{c prop: Chernoff bound}
(Chernoff bound, in \cite[page 69]{mitzenmacher2017probability}) Let $X\sim binomial(n,p)$. Then for every $\delta\in (0,1]$,
\begin{align} \label{c equ: chernoff X<}
    P(X \le (1-\delta)pn) < e^{\color{black}-\frac{1}{2}\delta^2pn},
\end{align} and
\begin{align} \label{c equ: chernoff X>}
    P(X \ge (1+\delta)pn) < e^{\color{black}-\frac{1}{3}\delta^2pn}.
\end{align}
\end{proposition}

{\subsection{Analysis of the bitcoin backbone protocol}}
\begin{definition}
{(Good event)} For all $0\le s < t$ and $0<\delta<\frac{1}{2}$, the \emph{$\delta$-good event with respect to time interval $(s,t]$} is
\begin{align}
    \E{s}{t} = \Ea{s}{t} \cap \Eb{s}{t} \cap \Ec{s}{t} \cap \Ed{s}{t}
\end{align}
where
\begin{align}
    \Ea{s}{t} & =  \left\{(1-\delta)(t-s)\alpha < \N{s}{t} < (1+\delta)(t-s)\alpha \right\}\label{c equ: E1} \\
    \Eb{s}{t} & = \left\{(1-\delta)(t-s)g\alpha < \X{s}{t} \right\}\label{c equ: E2} \\
    \Ec{s}{t} & = \left\{(1-\delta)(t-s)g^2\alpha < \Y{s}{t} \right\}\label{c equ: E3} \\
    \Ed{s}{t} & = \left\{\Z{s}{t} <(t-s)\beta + (t-s)g^2\alpha\delta\right\}. \label{c equ: E4}
\end{align}
\end{definition}
\begin{draft}

\end{draft}
Basically, under $\E{s}{t}$,  there exist 1) a ``typical'' number of honest blocks, 2) ``enough'' laggers and loners, and 3) not too many adversarial blocks.

\begin{lemma} \label{lemma: exp N}
For all real numbers $0\le s < t$,
\begin{align}
    \mathbb{E}[\N{s}{t}] = \alpha(t-s).
\end{align}
\end{lemma}

\begin{proof}
The result follows from the fact that $\N{s}{t}$ is a Poisson distribution with parameter $\alpha(t-s)$.
\end{proof}

\begin{lemma} \label{lemma: c prob Ea}
For all $0<\delta<\frac{1}{2}$ and $0\le s < t$,
\begin{align}
  P\left((\Ea{s}{t})^c\right) <  2e^{-\frac{1}{3}\delta^2\alpha (t-s)}.
\end{align}
\end{lemma}
\begin{proof}
\begin{align}
    P\left((\Ea{s}{t})^c\right)&  = P \left(\N{s}{t}\le (1-\delta)(t-s)\alpha\right) + P\left(\N{s}{t}\ge (1+\delta)(t-s)\alpha \right) \\
    & = P\left(\N{s}{t} \le \mathbb{E}[\N{s}{t}] - \delta \mathbb{E}[\N{s}{t}]\right) + P\left( \N{s}{t} \ge \mathbb{E}[\N{s}{t}] +\delta \mathbb{E}[\N{s}{t}]\right) \label{c 1.09}\\
    & <2e^{-\frac{1}{3}\delta^2\alpha (t-s)} \label{c 1.10},
\end{align}
where \eqref{c 1.09} is due to Lemma \ref{lemma: exp N} and \eqref{c 1.10} is due to Lemma \ref{lemma: Poisson property}.
\end{proof}

\begin{lemma} \label{lemma: Ui bernoulli}
Random variables $\U{1},\U{2},\ldots$ are independent Bernoulli random variables with
\begin{align}
P(\U{i} = 1) = g, \; i = 1,2,\ldots \label{c 0.912}
\end{align}
\end{lemma}
\begin{proof}
The inter-arrival times of the Poisson process $(N_t, t>0)$ are independent exponential random variables with the same parameter $\alpha$ \cite[Page 419]{ross1998first}. The probability a lag exceeds $\Delta$ is equal to $g$, hence \eqref{c 0.912} follows.
\end{proof}

\begin{lemma}\label{lemma: c exist int}
Suppose $0<\delta<\frac{1}{2}$ and $0\le s < t-\frac{80}{\delta}$. Then there exists a number $\delta_1\in [\frac{29}{60}\delta, \frac{31}{60}\delta]$ such that $(1-\delta_1)(t-s)\alpha$ is an integer.
\end{lemma}
\begin{proof}
Let function $f(x) = (1-x)(t-s)\alpha$. Obviously $f(x)$ is a continuous function on closed interval $[\frac{29}{60}\delta, \frac{31}{60}\delta]$. Note that
\begin{align}
    f(\frac{29}{60}\delta)-f(\frac{31}{60}\delta) & = \frac{\delta}{30}(t-s)\alpha \\
    & > \frac{8}{3}\alpha \label{c 0.913}\\
    & > 1   \label{c 0.914}
\end{align}
where \eqref{c 0.913} is due to $t-s>\frac{80}{\delta}$ and \eqref{c 0.914} is due to $\alpha > \frac{1}{2}$.
Let $n = \left\lceil f(\frac{31}{60}\delta) \right \rceil$.
Then we have $f(\frac{31}{60}\delta)\le n < f(\frac{29}{60}\delta)$. According to the Intermediate Value Theorem, there must exist a $\delta_1\in [\frac{29}{60}\delta, \frac{31}{60}\delta]$ such that $f(\delta_1) = (1-\delta_1)(t-s)\alpha = n$. Hence the proof.
\end{proof}

\begin{lemma} \label{lemma: c prob Eb}
For all $0<\delta<\frac{1}{2}$ and $0\le s < t-\frac{80}{\delta}$,
\begin{align}
 P\left((\Eb{s}{t})^c\right) < 2e^{-\frac{1}{12}(t-s)\delta^2g\alpha}.
\end{align}
\end{lemma}
\begin{proof}
According to Lemma \ref{lemma: c exist int}, there exists a $\delta_1 \in [\frac{29}{60}\delta, \frac{31}{60}\delta]$ such that $(1-\delta_1)(t-s)\alpha$ is an integer. Denote \begin{align}
    n = (1-\delta_1)(t-s)\alpha. \label{c 1.01}
\end{align} Let $\delta_2 = \delta - \delta_1$. Then $\delta_2\in[\frac{29}{60}\delta,\frac{31}{60}\delta]$.
Define events
\begin{align}
{K} & = \left\{\N{s}{t} \ge n\right\} \label{c 1.02}\\
{L} & = \left\{ \sum_{i=N_s+1}^{N_s+n}\U{i} \ge (1-\delta_2)ng \right\}. \label{c 1.03}
\end{align}
Note that under event ${K \cap L}$, we have
\begin{align}
    \X{s}{t} & = \sum_{i=N_s+1}^{N_t}\U{i} \label{c 1.04} \\
    & \ge  \sum_{i=N_s+1}^{N_s+n}\U{i} \label{c 1.05}\\
    & \ge (1-\delta_2)ng \label{c 1.06} \\
    & = (1-\delta_2)(1-\delta_1)(t-s)g\alpha \label{c 1.07}\\
    & > (1-\delta)(t-s)g\alpha \label{c 1.08}
\end{align}
where \eqref{c 1.04} is due to \eqref{equ: XU}, \eqref{c 1.05} is due to \eqref{c 1.02}, \eqref{c 1.06} is due to \eqref{c 1.03}, \eqref{c 1.07} is due to \eqref{c 1.01}, and \eqref{c 1.08} is due to $\delta_2 = \delta - \delta_1$.
By \eqref{c 1.08} we have
\begin{align}
    {K\cap L} \subset \Eb{s}{t}. \label{equ: c subset KE}
\end{align}
Note that
\begin{align}
    P\left({K^c}\right) & = P\left(\N{s}{t} < n \right) \\
    & =  P\left(\N{s}{t} < (1-\delta_1)(t-s)\alpha \right)\\
    & < e^{-\frac{1}{2}(t-s)\delta_1^2\alpha} \label{c 1.081} \\
    & \le e^{-\frac{1}{12}(t-s)\delta^2\alpha}  \label{c 1.0810}
\end{align}
where \eqref{c 1.081} is due to Proposition \ref{c prop: Chernoff bound} and \eqref{c 1.0810} is due to $\delta_1\ge \frac{29}{60}\delta$. Also,
\begin{align}
P\left({L^c}\right) & = P\left(  \sum_{i=N_s+1}^{N_s+n}\U{i} < (1-\delta_2)ng \right) \\
    & <e^{-\frac{1}{2}\delta_2^2 ng }\label{c 1.082} \\
    & =e^{-\frac{1}{2}(1-\delta_1)(t-s)\delta_2^2g\alpha}\label{c 1.083}\\
    & < e^{-\frac{1}{12}(t-s)\delta^2g^2\alpha} \label{c 1.085}
\end{align}
where \eqref{c 1.082} is due to Proposition \ref{c prop: Chernoff bound}, \eqref{c 1.083} is due to \eqref{c 1.01}, and \eqref{c 1.085} is due to $\delta<\frac{1}{2}$, $1-\delta_1 \ge 1-\frac{31}{60}\delta > \frac{89}{120}$, and $\delta_2 \ge \frac{29}{60}\delta$. Thus,
\begin{align}
    P\left((\Eb{s}{t})^c\right) &  \le P(({K \cap L})^c) \label{c 1.087}\\
    & \le P({K^c}) + P({L^c}) \\
    & < 2e^{-\frac{1}{12}(t-s)\delta^2g\alpha} \label{c 1.086}
\end{align}
where \eqref{c 1.087} is due to \eqref{equ: c subset KE} and \eqref{c 1.086} is due to \eqref{c 1.0810} and \eqref{c 1.085}.
\end{proof}

\begin{lemma} \label{lemma: Vi bernolli}
Random variables $\V{1},\V{3},\ldots$ are independent Bernoulli random variables with
\begin{align}
P(\V{i} = 1) = g^2, \; i = 1,3,\ldots \label{c 0.92}
\end{align}
Random variables  $\V{2},\V{4},\ldots$ are independent Bernoulli random variables with
\begin{align}
    P(\V{i}=1) = g^2, \; i = 2, 4, \ldots
\end{align}
\end{lemma}
\begin{proof}
Lemma \ref{lemma: Vi bernolli} follows from the fact that $\V{i} = \U{i}\U{i+1}$ and $\U{i}$s are independent of each other.
\end{proof}

\begin{lemma}\label{lemma: c prob Ec}
For all $0<\delta<\frac{1}{2}$ and $0\le s < t-\frac{80}{\delta}$,
\begin{align}
  P\left((\Ec{s}{t})^c\right) < 4e^{-\frac{1}{24}(t-s)\delta^2g^2\alpha}.
\end{align}
\end{lemma}
\begin{proof}
According to Lemma \ref{lemma: c exist int}, there exists a $\delta_1 \in [\frac{29}{60}\delta, \frac{31}{60}\delta]$ such that $(1-\delta_1)(t-s)\alpha$ is an integer. Denote \begin{align}
    n = (1-\delta_1)(t-s)\alpha. \label{c 2.01}
\end{align} Let $\delta_2 = \delta - \delta_1$. Then $\delta_2\in[\frac{29}{60}\delta,\frac{31}{60}\delta]$. Suppose there are $n_1$ even numbers and $n_2$ odd numbers in $\{1,\ldots,n\}$. Obviously $n_1+n_2=n$.
Define events
\begin{align}
K & = \left\{\N{s}{t} \ge n\right\} \label{c 2.02}\\
L & = \left\{
\V{N_s+2} + \V{N_s+4} + \ldots + \V{N_s+2n_1} \ge (1-\delta_2)g^2n_1
\right\} \label{c 2.03} \\
S & = \left\{
\V{N_s+1} + \V{N_s+3} + \ldots + \V{N_s+2n_2-1} \ge (1-\delta_2)g^2n_2
\right\}. \label{c 2.035}
\end{align}
Suppose event ${K \cap L \cap S}$ occurs, we have
\begin{align}
    \Y{s}{t} & = \sum_{i=N_s+1}^{N_t}\V{i} \label{c 2.04} \\
    &\ge  \sum_{i=N_s+1}^{N_s+n}\V{i} \label{c 2.05}\\
    & = (\V{N_s+2} + \V{N_s+4} + \ldots + \V{N_s+2n_1}) + (\V{N_s+1} + \V{N_s+3} + \ldots + \V{N_s+2n_2-1})\\
    & \ge  (1-\delta_2)g^2n_1 + (1-\delta_2)g^2n_2 \label{c 2.06}\\
    & = (1-\delta_2)g^2 n \\
    & = (1-\delta_2)(1-\delta_1)(t-s)g^2\alpha \label{c 2.07}\\
    & > (1-\delta)(t-s)g^2\alpha.\label{c 2.08}
\end{align}
where \eqref{c 2.04} is due to \eqref{equ: YV}, \eqref{c 2.05} is due to \eqref{c 2.02}, \eqref{c 2.06} is due to \eqref{c 2.03} and \eqref{c 2.035}, \eqref{c 2.07} is due to \eqref{c 2.01}, and \eqref{c 2.08} is due to $\delta_2 = \delta - \delta_1$.
By \eqref{c 2.08} we have
\begin{align}
   {K \cap L \cap S} \subset \Ec{s}{t}. \label{equ: c subset KES}
\end{align}
Note that
\begin{align}
    P\left({K^c}\right) & = P\left(\N{s}{t} < n \right) \\
    & =  P\left(\N{s}{t} < (1-\delta_1)(t-s)\alpha \right)\\
    & < e^{-\frac{1}{2}(t-s)\delta_1^2\alpha} \label{c 2.081} \\
    & \le e^{-\frac{1}{12}(t-s)\delta^2\alpha}  \label{c 2.0810}
\end{align}
where \eqref{c 2.081} is due to Proposition \ref{c prop: Chernoff bound} and \eqref{c 2.0810} is due to $\delta_1\ge \frac{29}{60}\delta$. Also,
\begin{align}
P\left({L^c}\right) & = P\left( \V{N_s+2} + \V{N_s+4} + \ldots + \V{N_s+2n_1} < (1-\delta_2)g^2n_1 \right) \\
    & < e^{-\frac{1}{2}\delta_2^2 n_1 g^2 }\label{c 2.082} \\
    & \le e^{-\frac{1}{4}\delta_2^2 (n-1) g^2} \label{c 20} \\
    & < \frac{3}{2}e^{-\frac{1}{4}\delta_2^2 n g^2} \label{c 20.1}\\
    & = \frac{3}{2}e^{-\frac{1}{4}(1-\delta_1)(t-s)\delta_2^2g^2\alpha}\label{c 2.083}\\
    & < \frac{3}{2}e^{-\frac{1}{24}(t-s)\delta^2{g^2}\alpha} \label{c 2.085}
\end{align}
where \eqref{c 2.082} is due to Proposition \ref{c prop: Chernoff bound}, \eqref{c 20} is due to $n_1\ge \frac{n-1}{2}$, \eqref{c 20.1} is due to $e^{\frac{\delta^2 g}{4}} < \frac{3}{2}$,
\eqref{c 2.083} is due to \eqref{c 2.01}, and \eqref{c 2.085} is due to $\delta<\frac{1}{2}$, $1-\delta_1 \ge 1-\frac{31}{60}\delta > \frac{89}{120}$, and $\delta_2 \ge \frac{29}{60}\delta$.  Similarly, we have
\begin{align}
    P\left({S^c}\right) & <  \frac{3}{2}e^{-\frac{1}{24}(t-s)\delta^2g^2\alpha}. \label{c 20.2}
\end{align}

Thus,
\begin{align}
    P\left((\Ec{s}{t})^c\right) &  \le P(({K \cap L}\cap S)^c) \label{c 2.087}\\
    & \le P({K^c}) + P({L^c}) +  P({S^c})\\
    & < 4e^{-\frac{1}{24}(t-s)\delta^2g^2\alpha} \label{c 2.086}
\end{align}
where \eqref{c 2.087} is due to \eqref{equ: c subset KES} and \eqref{c 2.086} is due to \eqref{c 2.0810}, \eqref{c 2.085} and \eqref{c 20.2}.
\end{proof}


\begin{lemma}\label{lemma: c prob Ed}
For all $0<\delta<\frac{1}{2}$ and $0\le s < t-\frac{80}{\delta}$,
\begin{align}
P\left((\Ed{s}{t})^c\right) < e^{-\frac{1}{6}(t-s)\delta^2g^2\alpha}.
\end{align}
\end{lemma}
\begin{proof}
Assume $\Z{s}{t}'$ is a Poisson distribution with parameter $\beta(t-s)$. We have
\begin{align}
P\left((\Ed{s}{t})^c\right)& = P\left(\Z{s}{t} \ge \beta(t-s) +  (t-s)g^2\delta \alpha\right)  \\
& \le P\left(\Z{s}{t}' \ge \beta(t-s) +  (t-s)g^2\delta \alpha\right) \label{c 2.-2}\\
& < \mathbb{E}[e^{u(\Z{s}{t}'-\beta(t-s) -  (t-s)g^2\delta\alpha)}] \\
& =  \frac{\mathbb{E}\left[e^{\Z{s}{t}'u}\right]}{e^{\beta(t-s)u + (t-s)g^2\delta\alpha}} \label{c 2.-1} \\
& = \frac{e^{\beta (t-s) (e^u -1)}}{e^{\beta (t-s) u + \delta g^2\alpha (t-s)u}} \label{c he: 2.0}\\
& \le  \frac{e^{\beta (t-s) (e^u -1)}}{e^{\beta (t-s) u +\frac{1}{2}\beta(t-s)u + \frac{1}{2}\delta g^2\alpha (t-s)u}} \label{c he: 2.01}\\
& =  e^{\left(e^u-1-u(1+\frac{\delta}{2})\right)\beta(t-s)-\frac{\delta}{2} g^2 \alpha (t-s)u}, \label{c he: 2.1}
\end{align}
where \eqref{c 2.-2} is due to \eqref{equ: c Z bound},  \eqref{c he: 2.0} is due to the fact that the moment generating function for a Poisson random variable with parameter $\lambda$ is $e^{\lambda (e^u - 1)}$, and \eqref{c he: 2.01} is due to $g^2\alpha > \beta$. Picking $u = \log(1+\frac{\delta}{2})$, we have
\begin{align}
P((\Ed{s}{t})^c) &  <  e^{\left(e^u-1-u(1+\frac{\delta}{2})\right)\beta(t-s)-\frac{\delta}{2} g^2 \alpha (t-s)u}, \\
& < e^{-\frac{\delta}{2}  \log(1+\frac{\delta}{2})g^2 \alpha (t-s)} \label{c 2.2}\\
& < e^{-\frac{1}{6}\delta^2 g^2\alpha (t-s)}\label{c 2.3}
\end{align}
where \eqref{c 2.2} is due to $\frac{\delta}{2} - (1+\frac{\delta}{2})\log(1+\frac{\delta}{2}) < 0$ for all $\delta \in (0,1)$ and \eqref{c 2.3} is due to $\log(1+\frac{\delta}{2}) > \frac{\delta}{3}$ for all $\delta \in (0,1)$.
\end{proof}

For convenience, let
\begin{align} \label{def: c eta}
    \eta = \delta^2g^2\alpha.
\end{align}

\begin{lemma}\label{lemma: c prob of good events}
For all $0<\delta<\frac{1}{2}$ and $0\le s < t-\frac{80}{\delta}$,
\begin{align}
    P\left(\E{s}{t}\right)  >1 - 9e^{-\frac{1}{24}\eta (t-s)}.
\end{align}
\end{lemma}

\begin{proof}
By Lemma \ref{lemma: c prob Ea},  Lemma \ref{lemma: c prob Eb},  Lemma \ref{lemma: c prob Ec}, and  Lemma \ref{lemma: c prob Ed}, we have
\begin{align}
P\left(\E{s}{t}\right) & = 1-P\left((E^{\delta }[s,t])^c\right) \\
& \ge 1- P\left((\Ea{s}{t})^c\right) - P\left((\Eb{s}{t})^c\right) - P\left((\Ec{s}{t})^c\right)-P\left((\Ed{s}{t})^c\right) \\
& > 1-9e^{-\frac{1}{24}\delta^2 g^2 \alpha (t-s)} \label{c 2.4} \\
& = 1-9e^{-\frac{1}{24}\eta(t-s)}.
\end{align}
\end{proof}
Intuitively, during a longer time interval the number of each type of blocks mined is more likely to be close to their expected value, thus the probability of $\E{s}{t}$ is higher.
Lemma \ref{lemma: c prob of good events} proves the probability that $\E{s}{t}$ does not occur vanishes exponentially with $t-s$.

\begin{definition} \label{def: c typical event}
(Typical event) The $\delta$-typical event on interval $(s,t]$ is defined as:
\begin{align}
    \G{s}{t} = \bigcap_{a\in [0,s], b \in [t,\infty)} \E{a}{b}. \label{equ: c def typical}
\end{align}
\end{definition}

\begin{definition}
For $0\le s < t$, define
\begin{align}
    \J{s}{t} =  \bigcap_{k\in \{0, \ldots \lceil s\rceil\}, \ell \in \{\lfloor t \rfloor, \lfloor t \rfloor + 1, \ldots \}}\E{k}{\ell}.
\end{align}
\end{definition}

Evidently, the $J$ event is the intersection of countably many events, whereas the $G$ event is the intersection of uncountably many events. The following relationship is important:

\begin{lemma} \label{lemma: c J in G}
For all real numbers $0\le s < t-\coeff$,
\begin{align}
    \Jth{s}{t} \subset \G{s}{t}.
\end{align}
\end{lemma}

\begin{draft}
\begin{proof}
We show that if $\Jth{s}{t}$ occurs, $\E{a}{b}$ occurs for all $a\in [0, s]$ and $b \in [t, \infty)$.

To prove $\Ea{a}{b}$ occurs, we have
\begin{align}
\N{a}{b} & \ge \N{\lceil a \rceil}{\lfloor b \rfloor} \\
& > (1-\coef) \alpha ( \lfloor b \rfloor-\lceil a \rceil) \label{c 2.20}\\
& >  (1-\coef) \alpha (b-a -2) \\
& = (1-\coef) \alpha (b-a)(1 - \frac{2}{b-a}) \\
& >  (1-\coef)(1-\coefff) \alpha (b-a) \label{c 2.22}\\
& > (1-\delta)(b-a)\alpha
\end{align}
where \eqref{c 2.20} is due to \eqref{c equ: E1} and {\eqref{c 2.22}} is due to $b-a>\coeff$.
Also,
\begin{align}
\N{a}{b} & < \N{\lfloor a\rfloor}{\lceil b \rceil}  \\
& = (1+\coef) \alpha (\lceil b \rceil - \lfloor a\rfloor) \label{c 2.30}\\
& < (1+\coef) \alpha (b-a + 2) \\
& = (1+\coef) \alpha (b-a)(1 + \frac{2}{b-a}) \\
& <  (1+\coef)(1+\coefff) \alpha (b-a)\label{c 2.32}\\
& < (1+\delta)(b-a)\alpha
\end{align}
where \eqref{c 2.30} is due to \eqref{c equ: E1} and {\eqref{c 2.32}} is due to $b-a>\coeff$.

To prove $\Eb{a}{b}$ occurs, we have
\begin{align}
\X{a}{b} & > \X{\lceil a \rceil}{\lfloor b \rfloor}\\
& = (1-\coef) g \alpha (\lfloor b \rfloor-\lceil a \rceil)  \label{c 2.40}\\
& >  (1-\coef) g\alpha (b-a -2) \\
& = (1-\coef) g\alpha (b-a)(1 - \frac{2}{b-a}) \\
& >  (1-\coef)(1-\coefff) g\alpha (b-a) \label{c 2.42}\\
& > (1-\delta)(b-a)g\alpha
\end{align}
where \eqref{c 2.40} is due to \eqref{c equ: E2} and {\eqref{c 2.42}} is due to $b-a>\coeff$.

To prove $\Ec{a}{b}$ occurs, we have
\begin{align}
\Y{a}{b} & > \Y{\lceil a \rceil}{\lfloor b \rfloor}   \\
& = (1-\coef) g^2 \alpha (\lfloor b \rfloor - Y\lceil a \rceil ) \label{c 2.50}\\
& >  (1-\coef) g^2 \alpha (b-a -2) \\
& = (1-\coef) g^2 \alpha (b-a)(1 - \frac{2}{b-a}) \\
& >  (1-\coef)(1-\coefff) g^2 \alpha (b-a)\label{c 2.52}\\
& > (1-\delta)(b-a)g^2\alpha
\end{align}
where \eqref{c 2.50} is due to \eqref{c equ: E3}and {\eqref{c 2.52}} is due to $b-a>\coeff$.

At last, to prove $\Ed{a}{b}$ occurs, we have
\begin{align}
\Z{a}{b} & < \Z{\lfloor a\rfloor}{\lceil b \rceil}\\
& =  \beta (\lceil b \rceil - \lfloor a\rfloor) + \coef  g^2\alpha  (\lceil  b \rceil - \lfloor a\rfloor) \label{c 2.60}  \\
& <  \beta (b-a + 2) + \coef  g^2\alpha  (b-a + 2) \\
& =  \beta (b-a)(1 + \frac{2}{b-a}) + \coef  g^2\alpha(b-a)(1 + \frac{2}{b-a}) \\
& < (1+\coefff)\beta(b-a) + (1+\coefff)\coef  g^2\alpha(b-a) \\
& < \beta(b-a) + \delta g^2\alpha(b-a) \label{c 2.64} \\
& = \mathbb{E}[\Z{a}{b}] + \delta (b-a)g^2\alpha
\end{align}
where \eqref{c 2.60} is due to \eqref{c equ: E4} and \eqref{c 2.64} is due to \eqref{equ: alpha beta}.

To sum up, under {$\Jth{s}{t}$},  events $\Ea{a}{b}$, $\Eb{a}{b}$, $\Ec{a}{b}$, and $\Ed{a}{b}$ occur for all  $a\in[0,s]$ and $b\in[t,\infty)$. Thus $\E{a}{b}$ occurs for all {those $a$ and $b$,  which implies that $\G{s}{t}$ also occurs.}
\end{proof}
\end{draft}

Bounding the probability of the $G$ event by the $J$ event allows us to use the union bound{, which eases} the calculation of the probability of typical events.


For convenience, we introduce the following parameter:
\begin{align} \label{def: c mu}
    \mu = \frac{9e^{\frac{2}{{27}}\eta}}{\left(1-e^{-\frac{1}{{27}}\eta}\right)^2}.
\end{align}

\begin{lemma}\label{lemma: c prob of typical event}
For all real numbers $0\le s < t-\coeff$,
\begin{align}
    P\left(\G{s}{t}\right) > 1-  \mu  e^{-\frac{1}{{27}}\eta(t-s)}.
\end{align}
\end{lemma}
\begin{draft}
\begin{proof}
By Lemma \ref{lemma: c J in G},
\begin{align}
    P\left((\G{s}{t})^c\right) & \le  P\left(({\Jth{s}{t}})^c\right) \\
    & = P\left( \bigcup_{k\in \{0, \ldots, \lceil s\rceil\}, \ell \in \{\lfloor t \rfloor, \lfloor t \rfloor + 1, \ldots \}}\left(\Eth{k}{\ell}\right)^c\right) \\
    & < \sum_{k\in \{0, \ldots, \lceil s\rceil\}, \ell \in \{\lfloor t \rfloor, \lfloor t \rfloor + 1, \ldots \}} 9e^{-\frac{1}{24}\left(\coef\right)^2g^2\alpha(\ell - k)} \\
    & < 9\left( \sum_{k=0}^{\lceil s\rceil} e^{\frac{1}{{27}}\eta k}  \right) \left( \sum_{\ell = \lfloor t \rfloor}^{\infty}e^{-\frac{1}{{27}}\eta \ell}\right)\\
    & = 9\frac{1-e^{\frac{1}{{27}}\eta(\lceil s\rceil + 1)}}{1-e^{\frac{1}{{27}}\eta}}\cdot\frac{e^{-\frac{1}{{27}}\eta \lfloor t \rfloor}}{1-e^{-\frac{1}{{27}}\eta}}\\
    & = 9\frac{\left( e^{\frac{1}{{27}}\eta \ceils} - e^{-\frac{1}{{27}}\eta}\right)e^{-\frac{1}{{27}}\eta \floort}}{\left(1-e^{-\frac{1}{{27}}\eta}\right)^2} \\
    & < 9\frac{\left( e^{\frac{1}{{27}}(s+1)\eta } - e^{-\frac{1}{{27}}\eta}\right)e^{-\frac{1}{{27}}(t-1)\eta }}{\left(1-e^{-\frac{1}{{27}}\eta}\right)^2} \\
    & < \mu e^{-\frac{1}{{27}}\eta(t-s)}.\label{c 2.8}
\end{align}
\end{proof}
\end{draft}

\begin{draft}
\begin{short}
\begin{proof}
Lemma \ref{lemma: c prob of typical event} can be proved based on and Lemma \ref{lemma: c J in G} and Lemma \ref{lemma: c prob of good events}.
\end{proof}
\end{short}
\end{draft}

\begin{lemma} \label{lemma: height concensus}
If a $t$-credible blockchain has height $h$, then the heights of all $(t+\Delta)$-credible blockchains are at least $h$.
\end{lemma}
\begin{proof}
Lemma \ref{lemma: height concensus} is obvious by the Definition \ref{def: c longest blockchain}.
\end{proof}

\begin{lemma} \label{lemma: c lagger diff height}
(Lemma 4 in \cite{Ren2019Analysis}) Laggers have different heights.
\end{lemma}
\begin{draft}
\begin{proof}
Suppose two laggers block $b$ and block $d$ with $T_d\ge T_b$ have the same height $k$.
Because block $d$ is a lagger, we must have {$T_d > T_b + \Delta$}.
According to Lemma \ref{lemma: honest block append to honest blockchain}, blockchain $b$ is {$T_b$}-credible. According to Lemma \ref{lemma: height concensus}, {the heights of} all {$(T_b+\Delta)$}-credible blocks are at least $k$. So the height of block $d$ is at least $k+1$, which contradicts the assumption.
\end{proof}
\end{draft}

\begin{lemma} \label{lemma: c unique kth block}
(Lemma 4 in \cite{Ren2019Analysis}) A loner is the only honest block at its height.
\end{lemma}
\begin{draft}

\begin{proof}
Suppose block $b$ mined at time $t$ is a loner.
By definition of a loner, no other honest block is mined during $[t-\Delta,t+\Delta]$.  Since blockchain $b$ is $t$-credible, the heights of all honest blocks mined after $t+\Delta$ must be at least ${\h{b}+1}$.
If an honest block is mined before $t-\Delta$, its height must be smaller than $\h{b}$ (if {its} height were $\h{b}$ or higher, block $b$'s height would be at least $\h{b}+1$).
\end{proof}
\end{draft}

\begin{definition}\label{def: c k deep}
($k$-deep block, $k$-deep prefix) Suppose {$k\in \{1,\ldots, n\}$}. By the $k$-deep block of blockchain $(\bl{0},\bl{1}, \ldots, \bl{n})$ we mean block {$\bl{n-k+1}$}. By the $k$-deep prefix of blockchain $(\bl{0},\bl{1}, \ldots, \bl{n})$ we mean blockchain $\bl{n-k}$.
\end{definition}
{By Definition \ref{def: c k deep}, a $k$-deep block extends a $k$-deep prefix }{of the same blockchain.}

\begin{lemma} \label{lemma: c growth}
Suppose positive integer $k$ and real number $t$ satisfy $t\ge \frac{k}{2\alpha}$. Under event $\E{t-\frac{k}{2\alpha}}{t}$, the   $(k-1)$-deep prefix of every blockchain mined by time $t$ must be mined no later than time $t-\frac{k}{2\alpha}$.
\end{lemma}
\begin{proof}
Under $\E{t-\frac{k}{2\alpha}}{t}$,
the total number of blocks mined by all miners during $(t-\frac{k}{2\alpha}, t]$ is upper bounded by
\begin{align}
    \N{t-\frac{k}{2\alpha}}{t}+\Z{t-\frac{k}{2\alpha}}{t}
    & < (1+\delta)\alpha \frac{k}{2\alpha} + \beta \frac{k}{2\alpha} + \delta g^2\alpha \frac{k}{2\alpha} \label{c 3.0}  \\
    & < (1+\delta)\alpha \frac{k}{2\alpha} + \left(1-\frac{81}{40}\delta\right)g^2\alpha \frac{k}{2\alpha} + \delta g^2\alpha \frac{k}{2\alpha} \label{c 3.005} \\
    & < (1+\delta)\alpha \frac{k}{2\alpha} + \left(1-\frac{41}{40}\delta\right)\alpha \frac{k}{2\alpha} \label{c 3.01}\\
    & < k \label{c 3.02}
\end{align}
where \eqref{c 3.0} is due to \eqref{c equ: E1} and \eqref{c equ: E3}, \eqref{c 3.005} is due to \eqref{equ: alpha beta}, and \eqref{c 3.01} is due to $g<1$. So the number of block mined during $(t-\frac{k}{2\alpha}, t]$ is at most $k-1$. Therefore,  the  $(k-1)$-deep prefix of every blockchain mined by time $t$ must be mined no later than time $t-\frac{k}{2\alpha}$.
\end{proof}

\begin{theorem} \label{thm: c growth}
(Blockchain growth theorem) Suppose real numbers $s$ and $t$ satisfy $0\le s < t-\coeft$. Under event $\E{s+\Delta}{t-\Delta}$,
the height of every $t$-credible blockchain is at least  $\growth g\alpha (t-s)$ larger than the maximum height of all $s$-credible blockchains. As a consequence, the probability that some $t$-credible blockchain is less than $\growth g\alpha(t-s)$ {higher} than some $s$-credible blockchain does not exceed $9e^{-\frac{1}{24}\eta(t-s)}$.
\end{theorem}
\begin{proof}
\begin{draft}
Assume {the maximum height of all} $s$-credible blockchains {is} $\ell$. According to Lemma \ref{lemma: height concensus}, the heights of all $(s+\Delta)$-credible blockchains are at least $\ell$.

Under event $\E{s+\Delta}{t-\Delta}$, during time interval $(s+\Delta, t-\Delta]$ the number of laggers is lower bounded:
\begin{align}
\X{s+\Delta}{t-\Delta} & > (1-\delta)g\alpha (t-s-2\Delta) \\
& = (1-\delta)g\alpha \frac{t-s-2\Delta}{t-s} (t-s) \\
&> (1-\delta)\left(1-\frac{\delta}{40}\right)g\alpha(t-s) \label{c 3.03}\\
& > \growth g\alpha(t-s)
\end{align}
where \eqref{c 3.03} is due to $t-s>\coeft$. According to Lemma \ref{lemma: c lagger diff height},
these laggers have different heights, thus there {must exist} a lagger with height of at least  $\ell + \growth g\alpha(t-s)$ {by} time $t-\Delta$. This height lower bounds the {heights} of all $t$-credible blockchains.

By Lemma \ref{lemma: c prob of good events}, the probability that some $t$-credible blockchain is less than $\growth g\alpha(t-s)$ higher than some $s$-credible blockchain does not exceed $9e^{-\frac{1}{24}\eta(t-s)}$.
\end{draft}

\begin{short}
Theorem \ref{thm: c growth} can be derived by \eqref{c equ: E2} and Lemma \ref{lemma: c lagger diff height}.
\end{short}
\end{proof}


\begin{lemma} \label{lemma: c t-credible height}
Suppose positive integer $k$ and real number $t$ satisfy $t \ge \frac{k}{\growth g\alpha}$. Then under event {$\G{t-\frac{k}{2\alpha}+\Delta}{t-\Delta}$}, the height of {every} $t$-credible blockchain is {at least} $k$.
\end{lemma}
\begin{proof}
{Let $r = \frac{k}{\growth g\alpha}$. Evidently, $r>\frac{k}{2\alpha}$.}
If event
{$\G{t-\frac{k}{2\alpha}+\Delta}{t-\Delta}$ } occurs, event  $\E{t-r+\Delta}{t-\Delta}$ occurs.
According to Theorem \ref{thm: c growth},
every $t$-credible blockchain is at least $\growth g\alpha r = k$ higher than the maximum height of all $\left(t-r\right)$-credible blockchains. Hence the proof.
\end{proof}

\begin{theorem}\label{thm: c quality}
(Blockchain quality theorem)
Suppose positive integer $k$ and real number $t$ satisfy {$k \ge \coefk$} and $t \ge \frac{k}{\growth g\alpha}$.
Under event $\G{t-\frac{k}{2\alpha}+\Delta}{t-\Delta}$, at least $\left(1-g\right)k$ of the last $k$ blocks of every $t$-credible blockchain are honest. As a consequence, the probability that more than  $gk$ of the last $k$ blocks of some $t$-credible blockchain are adversarial does not exceed $\mu {e^{-\frac{\eta}{{27}}\left(\frac{k}{2\alpha}-2\Delta\right)}}$.
\end{theorem}


\begin{proof}
The intuition is that under good events,
the heights of credible blockchains grow by at least $\X{s+\Delta}{t-\Delta}$ during $(s,t]$, which is lower bounded by \eqref{c equ: E2}. Meanwhile, the number of adversarial blocks mined is upper bounded by \eqref{c equ: E4}.
Thus, at least a fraction of blocks must be honest even in the worst case that all adversarial blocks are included in a credible blockchain.
\begin{short}
Theorem \ref{thm: c quality} can be proved based on Theorem \ref{thm: c growth}.
\end{short}

\begin{draft}
To be precise,
suppose blockchain $d$ is $t$-credible. According to Lemma \ref{lemma: c t-credible height}, under $\G{t-\frac{k}{2\alpha}+\Delta}{t-\Delta}$ we have $h(d)\ge k$. Denote the $k$-deep block of blockchain $d$ as block $b$.
Let block $e$ be the highest honest block mined before block $b$ on blockchain $b$. Then we have $0\le \h{e} < \h{b}$.
\DG{The relationship between these blocks is illustrated as follows:
\begin{align}
k \text{ blocks}  \qquad \nonumber \\
\fbox{$e$}-\cdots-\fbox{\color{white}i}-\overbrace{\fbox{$b$}-\cdots-\fbox{$d$}} \quad\\
\text{time } s \qquad \qquad \qquad \qquad \qquad \text{time } t
\nonumber
\end{align}}

{It is easy to check that $t>\frac{k}{2\alpha}$. Let $s = T_e$ for convenience.}
According to Lemma \ref{lemma: c growth}, we have {$t-s>  \frac{k}{2\alpha} > \coeft$} under event $\E{t-\frac{k}{2\alpha}}{t}$.

Denote the number of adversarial blocks between block $b$ (inclusive) and block $d$ (inclusive) as $z$.
By definition, on blockchain $d$, all blocks at heights  $\{\h{e}+1, \dots, \h{b}-1\}$ are adversarial. Let $y = \h{b}-\h{e}-1$, then $z + y$ is the number of adversarial blocks between block $e$ (exclusive) and the block $d$ (inclusive). Obviously these adversarial blocks must be mined during $(s, t]$, thus we have
\begin{align}
    z+ y & \le \Z{s}{t}  \\
    & < \beta(t-s) + \delta g^2 \alpha(t-s)\label{c 4.490}
\end{align}
where \eqref{c 4.490} is due to \eqref{c equ: E4}.

By Lemma \ref{lemma: honest block append to honest blockchain}, blockchain $e$ is  $s$-credible because block $e$ is honest. Blockchain $d$ is $t$-credible by definition.
By Theorem \ref{thm: c growth}, under $\E{s+\Delta}{t-\Delta}$,  a $t$-credible blockchain is more than  $\growth g\alpha(r-s)$ longer than an $s$-credible blockchain. Since the height difference of
blockchain $e$ and blockchain $d$ is $k+y$, we have
\begin{align}
   k+ y & \ge  \growth g\alpha(t-s).  \label{c 4.491}
\end{align}

Thus, {under event $\G{t-\frac{k}{2\alpha}+\Delta}{t-\Delta}$}, we have
\begin{align}
    \frac{z}{k} & \le \frac{z+ y}{k+y} \label{c 4.50} \\
    & < \frac{\beta(t-s)+\delta g^2\alpha(t-s)}{\growth  g\alpha (t-s)} \label{c 4.52} \\
    & = \frac{1}{\growth }\left( \frac{\beta}{g\alpha} + \delta g\right) \label{c 4.53} \\
    & < \frac{1}{\growth }\left( \left(1-\frac{81}{40}\delta\right)g + \delta g\right) \label{c 4.54} \\
    & = g, \label{c 4.55}
\end{align}
where \eqref{c 4.50} is due to $z \le k$,  \eqref{c 4.52} is due to \eqref{c 4.490} and \eqref{c 4.491}, and \eqref{c 4.54} is due to \eqref{equ: alpha beta}.

{According to Lemma \ref{lemma: c prob of typical event},
\begin{align}
{P\left(\G{t-\frac{k}{2\alpha}+\Delta}{t-\Delta}\right)} > 1-\mu {e^{-\frac{\eta}{{27}}\left(\frac{k}{2\alpha}-2\Delta\right)}}.
\end{align}

As a consequence, with probability at least $1-\mu {e^{-\frac{\eta}{{27}}\left(\frac{k}{2\alpha}-2\Delta\right)}}$, event {$\G{t-\frac{k}{2\alpha}+\Delta}{t-\Delta}$} occurs, under which the fraction of adversarial blocks in {the last} $k$  blocks of any $t$-credible blockchains is at most $g$. Hence the proof.}
\end{draft}

\end{proof}

\begin{definition}\label{def: c permanent}
A block or a sequence (of blocks) is said to be {\em permanent after $t$} if the block or sequence remains in all $r$-credible blockchains with $r\ge t$.
\end{definition}


\begin{lemma}\label{lemma: c Y > Z}
Suppose real numbers $s$, $t$, $r$, $\epsilon$ and integer $k$ satisfy  $k\ge \coefk$, $0 < s \le t-\frac{k}{2\alpha} < t \le r$,  and $0\le \epsilon < \frac{\delta}{40}(r-t)$. Then under event $\G{t-\frac{k}{2\alpha}+\Delta}{t-\Delta}$, we have
\begin{align}
   \Y{s +\Delta}{r-\epsilon-\Delta} > \Z{s}{r}.
\end{align}
\end{lemma}
\begin{proof}
{By assumption, it is easy to verify that}
\begin{align}
    \epsilon & < \frac{\delta}{40}(r-t) \\
    & <  \frac{\delta}{40}(r-s + s - t) \\
    & < \frac{\delta}{40}(r-s) - 2\Delta. \label{c 4.91}
\end{align}

{If $\G{t-\frac{k}{2\alpha}+\Delta}{t-\Delta}$ occurs, then $\E{s}{r}$ and $\E{s+\Delta}{r-\epsilon-\Delta}$ both occur, under which we have}
\begin{align}
    \Y{s +\Delta}{r-\epsilon-\Delta} & > (1-\delta)g^2\alpha(r-\epsilon-s-2\Delta) \label{c 5.0}\\
    & = (1-\delta)\frac{r-\epsilon-s-2\Delta}{r-s} g^2\alpha(r-s) \label{c 5.1}\\
    & > (1-\delta)\left(1-\frac{1}{40}\delta\right)g^2\alpha(r-s) \label{c 5.2}\\
    & > \left(1-\frac{81}{40}\delta\right)g^2\alpha(r-s) + \delta g^2\alpha (r-s)\\
    & > \beta (r-s)  + \delta g^2 \alpha (r-s)\label{c 5.3} \\
    & > \Z{s}{r} \label{c 5.4}
\end{align}
where \eqref{c 5.0} is due to \eqref{c equ: E3}, \eqref{c 5.2} is due to \eqref{c 4.91}, \eqref{c 5.3} is due to \eqref{equ: alpha beta}, and \eqref{c 5.4} is due to \eqref{c equ: E4}.
\end{proof}

\DG{
\begin{lemma} \label{lemma: c Y < Z}
  Suppose real numbers $s\le t \le r$.  If the highest honest block shared by an $r$-credible blockchain and a $t$-credible blockchain is mined at time $s$, then
  \begin{align} \label{eq: c Y < Z}
    \Y{s+\Delta}{t-\Delta} \le \Z{s}{r}.
\end{align}
\end{lemma}}


\begin{proof}
Suppose the highest block shared by $t$-credible  blockchain $d$ and $r$-credible blockchain $d'$ is block $b$.
Denote the highest honest block on blockchain $b$ \DG{as} block $e$ with $T_e = s$. Block $b$ and block $e$ may or may not be the same. 
\DG{The relationship between these blocks is illustrated as follows:
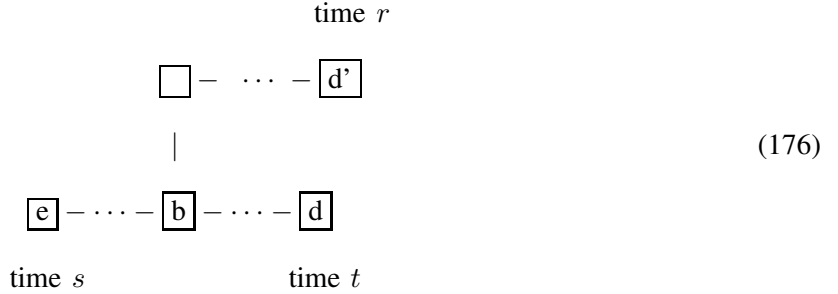
\begin{align}
\setlength{\jot}{0pt}
\text{time } r \nonumber \\
\fbox{\color{white}a}-\;\;\cdots\;-\fbox{d'} \quad \nonumber\\
|\quad\;  \qquad\qquad\qquad \\
\fbox{e}-\cdots-\fbox{b}-\cdots-\fbox{d} \qquad \nonumber\\
\text{time } s \quad \qquad \qquad \qquad \text{time } t \quad
\nonumber
\end{align}}

If $t-s\le 2\Delta$ or no loner is mined during $(s+\Delta, t-\Delta]$, obviously $ \Y{{s}+\Delta}{t-\Delta} \le \Z{{s}}{r}$.
Otherwise, consider loner $c$ mined during $({s}+\Delta, t-\Delta]$.
Since blockchain $e$ is ${s}$-credible and block $c$ is mined after time ${s}+\Delta$,  we have $\h{c} \ge \h{e}$ by Lemma \ref{lemma: height concensus}. Since blockchain $d$ is $t$-credible and blockchain  $d'$ is $r$-credible, we have $\h{c} \le \min\{\h{d}, \h{d'}\}$.  \DG{Consider the following two only possible cases:}
\begin{enumerate}
\item If $\h{e} < \h{c} \le \h{b}$, there exists at least one adversarial block at height $h(c)$ because all blocks between block $e$ (exclusive) and block $b$ (inclusive) are adversarial by definition.
\item If $\h{b} < \h{c} \le \min\{\h{d}, \h{d'}\}$,
there is at least one adversarial block at height $h(c)$, because two diverging blockchains exist but loner $c$ is the only honest block at its height by Lemma \ref{lemma: c unique kth block}.
\end{enumerate}
Thus, for \DG{every} loner $c$ mined during $({s}+\Delta, \DG{t}-\Delta]$, at least one adversarial block must be mined during $({s},r]$ at the same height.
\DG{In particular, the adversarial block must be mined before $r$ because it is published by time $r$.
Thus~\eqref{eq: c Y < Z} is proved.}
\end{proof}

\begin{lemma}\label{lemma: c common prefix lemma}
Suppose positive integer $k$ and real number $t$ satisfy $k \ge \frac{160\alpha(1+\Delta)}{\delta}$ and $t\ge \frac{k}{\growth g\alpha}$. Under event $\G{t-\frac{k}{2\alpha}+\Delta}{t-\Delta}$, the {$(k-1)$}-deep prefix of every $t$-credible blockchain  must be a prefix of all other $t$-credible blockchains (it may or may not be exactly {$(k-1)$}-deep in other $t$-credible blockchains).
\end{lemma}
\begin{proof}
If there exists only one $t$-credible blockchain, then the claim holds. Suppose there are multiple $t$-credible blockchains. Let block $e$ be the highest honest block shared by two of those blockchains: blockchain $d$ and blockchain $d'$.
For every $s\le t-\frac{k}{2\alpha}$, by Lemma \ref{lemma: c Y > Z} with $\epsilon=0$, we have $\Y{s + \Delta}{t-\Delta} > \Z{s}{t}$.
\DG{Meanwhile}, we have $\Y{T_e + \Delta}{t-\Delta} \le \Z{T_e}{t}$ by Lemma \DG{\ref{lemma: c Y < Z} with $s=T_e$ and $r=t$}.
Thus we must have $T_e > t - \frac{k}{2\alpha}$.
\DG{By Lemma \ref{lemma: c growth}, block $e$ cannot be on the {$(k-1)$}-deep prefix of blockchain $d$ or blockchain $d'$.  Thus, those two}
blockchains must diverge after their {$(k-1)$}-deep prefixes.
Hence the proof \DG{of this lemma}.
\end{proof}


\begin{theorem} \label{thm: c common prefix}
(Common prefix theorem)
Suppose positive integer $k$ and real number $t$ satisfy $k \ge \frac{160\alpha(1+\Delta)}{\delta}$ and $t\ge \frac{k}{\growth g\alpha}$. Under event $\G{t-\frac{k}{2\alpha}+\Delta}{t-\Delta}$, the {$(k-1)$}-deep prefix of a $t$-credible blockchain is extended by all {$r$}-credible blockchains with {$r\ge t$}. As a consequence, with probability at least  $1-\mu {e^{-\frac{\eta}{{27}}\left(\frac{k}{2\alpha}-2\Delta\right)}}$,
the {$(k-1)$}-deep prefix is permanent after time $t$. \end{theorem}

\begin{proof}

Let blockchain $b$ be the $(k-1)$-deep prefix of a $t$-credible blockchain. According to Lemma \ref{lemma: c common prefix lemma}, under event $\G{t-\frac{k}{2\alpha}+\Delta}{t-\Delta}$, all $t$-credible blockchains extend blockchain $b$ \DG{(i.e., the claim holds for $r=t$)}. We will show contradiction if blockchain $b$ is not permanent {after $t$}.

Contrary to the claim, assume there exists a credible blockchain after time $t$ which does not extend blockchain $b$.
Suppose $r$ is the smallest number in $(t, \infty)$ such that there
exists an $r$-credible blockchain (denoted as blockchain $d$) which does not \DG{extend} blockchain $b$. Then  all $(r-\epsilon)$-credible blockchains extend blockchain $b$ for $\epsilon \in (0, r-t]$.
Pick $\epsilon = \frac{\delta}{80}(r-t)$.
Denote one of the $(r-\epsilon)$-credible blockchains as blockchain $d'$.
Let block $e$ be the highest honest block shared by blockchain $d$ and blockchain $d'$. Let $T_e = s$ for convenience. We have $s < t-\frac{k}{2\alpha}$ by Lemma \ref{lemma: c growth}.
Note that $\epsilon < \frac{\delta}{40}(r-t)$.
According to Lemma \ref{lemma: c Y > Z}, under event $\G{t-\frac{k}{2\alpha}+\Delta}{t-\Delta}$ we have $\Y{s + \Delta}{r-\epsilon-\Delta} > \Z{s}{r}$.
On the other hand, by Lemma \ref{lemma: c Y < Z} we have $\Y{s + \Delta}{r-\epsilon-\Delta} \le \Z{s}{r}$. Contradiction arises. Thus under event $\G{t-\frac{k}{2\alpha}+\Delta}{t-\Delta}$, there can not exist any credible blockchain after time $t$ which does not extend blockchain $b$.

According to Lemma \ref{lemma: c prob of typical event},
\begin{align}
P\left(\G{t-\frac{k}{2\alpha}+\Delta}{t-\Delta}\right) > 1-\mu e^{-\frac{\eta}{27}\left(\frac{k}{2\alpha}-2\Delta\right)}.
\end{align}
As a consequence, the probability that there exist some credible blockchain after time $t$ which do not extend the $k$-deep prefix is does not exceed $\mu  e^{-\frac{\eta}{9}\left(\frac{k}{2\alpha}-2\Delta\right)}$.
\end{proof}

{Some other authors have developed a technique to prove results like Theorem \ref{thm: c common prefix}, which} defines a time interval $[t_0, t_1]$ where $t_1$ is the first time some honest miner adopts a blockchain $d'$ with a different $(k-1)$-deep prefix from blockchain $d$, and $t_0$ is the mining time of the highest honest block shared by blockchains $d$ and $d'$. Desired properties {are then claimed} under a good event with respect to interval $[t_0, t_1]$. However, such a proof is flawed because by picking $t_0$ and $t_1$ based on miners' behavior, the posterior statistics of the mining processes during $[t_0, t_1]$ are different than the prior statistics.

In this section, we have defined good events, studied the properties of various blocks under these events, and {provided bounds for their probabilities. Using these tools, we have proved} the blockchain growth theorem, the blockchain quality theorems, and the common prefix theorem, which guarantee the liveness and consistency of bitcoin blockchains. In essence, {a} bitcoin transaction deep enough in any {credible} blockchain is with high probability guaranteed to remain in the {transaction} ledger.

As a numerical example, consider a bitcoin payment system where the block propagation delay is upper bounded by $\Delta = 2$ seconds. Consistent with the bitcoin protocol, assume the arrival rate of honest blocks is $\alpha = {6}$ blocks per {hour} {$= \frac{1}{600}$ blocks per second}. Assume there is up to $25\%$ of adversarial mining power, thus $\beta \le {2}$ blocks per {hour}. We pick $\delta = 0.3285$ which satisfies \eqref{equ: alpha beta}.  Then $\eta=0.65$ and $\mu = 201.8$ by {\eqref{def: c eta} and \eqref{def: c mu}}.

Suppose a user broadcast a transaction $tx$ which is collected by block $b$. Let $k=26000$, then
\begin{align}
 1-\mu {e^{-\frac{\eta}{{27}}\left(\frac{k}{2\alpha}-2\Delta\right)}} > 1-10^{-20}.
\end{align}
It is easy to check that $k > 3000 > \coefk$ is satisfied.
According to Theorem \ref{thm: c common prefix}, if the user observes block $b$ is $26000$-deep, then with probability at least $1-\mu {e^{-\frac{\eta}{{27}}\left(\frac{k}{2\alpha}-2\Delta\right)}} > 1-10^{-20}$, transaction $tx$ is permanent.
{This guarantee 
may be improved by tightening the bounds developed in this section. This is left to future work.}

\section{The Prism backbone protocol}
The Prism protocol was invented and fully described in \cite{bagaria2019prism}. Here we describe the Prism  backbone {protocol} with just enough details to facilitate its analysis.

{\subsection{Model}}
Blocks are generated in a peer-to-peer network where honest and adversarial miners mine and publish blocks over time. The blocks are classified into $m+1$ categories, referred to as $0$-blocks,  $1$-blocks, \ldots, $m$-blocks.
A block is mined before knowing which kind of block it is, {so it contains} enough information for all $(m+1)$ kinds of blocks. Sortition relies on the range the new block's hash lands in: If a miner constructs a new block whose hash is within $[j\gamma, j\gamma + \gamma)$ for $j\in \{0,\ldots,m\}$, the mined block is a $j$-block.
Parameter $\gamma$ can be adjusted to control the mining rate.
All $j$-blocks mined by honest miners are called honest $j$-blocks. All $j$-blocks mined by adversaries are called adversarial $j$-blocks.

\begin{definition}
(The mining processes of j-blocks) For $j\in \{0, \ldots, m\}$, we assume  an honest genesis $j$-block, referred to as $j$-block $0$, is mined at time $0$.
Subsequent $j$-blocks are referred to as $j$-block $1$, $j$-block $2$, and so on, in the order they are mined in time after time $0$.
For $t>0$, let $\Mj{t}$ denote the
total number of non-genesis $j$-blocks mined by time $t$.
If a single $j$-block is mined at time $t$, it must be $j$-block $\Mj{t}$. If $k>1$ $j$-blocks are mined at exactly the same time $t$, we assume the tie is broken in some deterministic manner so that their block numbers are $\Mj{t}-k+1, \ldots, \Mj{t}$, respectively.
\end{definition}

\begin{definition}
($j$-blockchain) For $j\in \{0, \ldots, m\}$, every non-genesis $j$-block must contain the hash value of a unique parent $j$-block which is mined strictly earlier.  We use {$\fj{k} \in \{0,1,...,k-1\}$} to denote $j$-block $k$'s parent {$j$-block} number.
For $j$-block $k$ to be valid, there must exist a unique sequence of $j$-block numbers $\bl{0},\bl{1}, \dots, \bl{n}$ where $\bl{0}=0$, $\bl{n}=k$, and {$\fj{\bl{i}}=\bl{i-1}$} for $i =1,\ldots ,n$.  This sequence is referred to as $j$-blockchain $(\bl{0},\ldots,\bl{n})$ or simply $j$-blockchain $k$ since it is determined by $k$.
\end{definition}
We assume that $j$-block $k$ is validated {by} the existence of the entire $j$-blockchain $k$: To validate a $j$-block $k$, one needs access to the entire $j$-blockchain $k$. {We let $T^j_k$  denote the time when $j$-block $k$ is mined.   We let $P^j_k$ denote the time when $j$-block $k$ is published.}
\begin{definition}
{{(Height)} The height of $j$-block $k$, denoted as $\hj{k}$, is defined as the height of $j$-blockchain $k$, which is in turn defined as the number of non-genesis $j$-blocks in it.}
\end{definition}

\begin{definition} \label{def: cj credible voter blockchain}
($t$-credible {$j$-blockchains}) For $j\in\{0, \ldots,m\}$, we say $j$-blockchain $b$ is \emph{$t$-credible} if the $j$-blockchain has been published by time $t$, and is no shorter than any $j$-blockchain published by time $t-\Delta$.
That is to say,
\begin{align}
{P^j_{b}} &  \le t,
\end{align}
and
\begin{align}
{\hj{b}} & \ge \hj{k}, \quad \forall k: P^j_{k} \le t - \Delta.\label{equ: cj height of longest blockchain}
\end{align}
If there is no need to specify $t$ explicitly, $j$-blockchain $b$ can also be simply called a credible $j$-blockchain.
\end{definition}

For $j\in \{0,\ldots, m\}$, let $N^j_t$ denote the total number of honest $j$-blocks mined during $(0, t]$. We assume the sum mining rate of honest miners is $(m+1)\alpha$. Since the sortition scheme ensures the mining power of both honest
and adversarial miners evenly distributed across $j$-blockchains,
$(N^j_{t}, t\ge 0)$ is {an independent} homogeneous Poisson point process with rate $\alpha$.
Let $Z^j_t$ denote the total number of adversarial blocks mined {during $(0, t]$}, then $N^j_t + Z^j_t = M^j_t$.
We assume the sum mining rate of all adversaries is no larger than $(m+1)\beta$ and that the number of adversarial blocks mined during any time interval is upper bounded
probabilistically. Specifically, for real number $a$ and $0\le s<t$,
\begin{align}
    P\left(Z^j_{t}-Z^j_s\le a\right) \ge e^{-\beta(t-s)} \sum_{i=0}^{\lfloor a \rfloor}\frac{\left(\beta(t-s)\right)^i}{i!}. \label{equ: cj Z bound}
\end{align}
Note that the definition of a (credible) $j$-blockchain is identical to that of a (credible) bitcoin blockchain.  Thus, the blockchain growth theorem, blockchain quality theorem, common prefix theorem, and other properties of bitcoin blockchains remain valid in all $j$-blockchains.

However, it does not suffice to generate a high-throughput transaction ledger by simply putting $(m+1)$ bitcoin blockchains in parallel. In particular, while all transactions in each blockchain itself are consistent, transactions on different blockchains may contradict each other, {e.g., there may be double spending across different blockchains.}
In Prism protocol, these $(m+1)$ bitcoin blockchains {are} building block{s}.
An additional process, referred to as voting, is executed {to resolve conflicts and achieve} global consensus.
To be specific,
block are classified into {proposer} blocks ($0$-blocks) and {voter} blocks (all $j$-blocks with $j\in\{1,\ldots, m\}$).
Blockchains are classified into proposer blockchains ($0$-blockchains) and voter blockchains (all $j$-blockchains with $j\in\{1,\ldots, m\}$).
The voting {by credible voter blockchains} elects a series of proposer blocks called a {credible leader sequence}, which is responsible for generating a final transaction ledger. A {credible} leader sequence may or may not be a credible $0$-blockchain.
While the properties of bitcoin blockchains ensure the liveness and consistency of voter blockchains, the voting process ensures the liveness and consistency of {credible} leader sequences. Below we {briefly describe} voting and {transaction} ledger generation.

\begin{definition} \label{def: R_h}
For positive integer $h$, we let $R_h$ denote the time when the first proposer block on height $h$ is published.
\end{definition}

By saying a voter  $j$-block $b$ votes on a height $h$, we mean {the} voter block chooses one proposer block among all proposer blocks at height $h$ and points to {the proposer block} with a reference link. The reference link is part of the content of  voter {$j$-block} $b$, thus it is immutable. Obviously voter {$j$-block} $b$ can not vote on height $h$ with {$R_h \ge T^j_b$}.

According to the Prism protocol, when voting on a height, an honest voter block always {chooses the first observed proposer} block of this height.
An honest voter {$j$-block} $b$ votes on all heights $h$ as long as 1)  $R_h \le {T^j_b}-\Delta$ and 2) height $h$ has not been voted on by {the voter block's} ancestors.
An adversarial voter block may not choose the {first observed proposer} block when voting.
An adversarial voter block may refuse to vote on some height or repeatedly vote on some height that has already been voted by its ancestors.

At each height, the vote(s) from one voter blockchain is counted only once (only the first vote is valid if there exist several). In other words, proposer blocks on the same height  receive {up to} $m$ votes from $m$ voter blockchains in total.


\begin{definition} \label{def: cp credible leader sequence}
($t$-credible leader sequence)
Let proposer block $b_1,\ldots, b_n$ be proposer blocks at heights $1,\ldots, n$, respectively, where $n$ is greater or equal to the maximum height of all proposer blockchains published by time $t-\Delta$.
We say $({0, b_1}, \ldots, b_n)$ is a $t$-credible leader sequence if it is elected by a collection of $m$ $t$-credible voter blockchains including one $j$-blockchain for every $j\in \{1,\ldots,m\}$. That is, for every $\ell\in \{1,\ldots, n\}$,  proposer block $b_\ell$ receives the most votes among all proposer blocks of height $\ell$ at time $t$ from that collection of voter blockchains.
In particular, we have $n\ge h$ if height $h$ satisfies $R_h \le t-\Delta.$
{Block $b_\ell$} is called a $t$-credible leader block.
If there is no need to specify $t$ explicitly, {a} $t$-credible leader sequence {(block)} can also be simply {referred to as} a credible leader sequence {(block)}.
\end{definition}
{Note} that even if  all proposer blocks of height $h$ {have received} zero vote, {a} credible leader block of height $h$ can still exist according to the tie breaking rule. Moreover, {as} a $t$-credible leader sequence is defined with respect to a collection of $t$-credible voter blockchains,  in general there {can be} multiple $t$-credible leader sequences.
\begin{definition} \label{def: c observable}
{(Observable) Suppose block $b$ and block $d$ are published at time $t$ and time $s$, respectively. If $s\le t-\Delta$, we say block $d$ is {\em observable} from block $b$. If $s\ge t$, we say block $d$ is not observable from block $b$. If $t-\Delta<s<t$, block $d$ may or may not be observable from block $b$.}
\end{definition}

{Every} $t$-credible leader sequence {determines} a {transaction} ledger at time $t$.
According to the Prism protocol, as part of its content, an honest proposer block $b$ includes a reference link to every  proposer and voter block {that is observable from it and}  has not been pointed to by other reference links.
\begin{definition} \label{def: c reachable}
({Reachable}) By saying block $d$ is reachable from block $b$ (or block $b$ reaches block $d$), we mean block $b$ {points} to block $d$ by a sequence of reference links.
\end{definition}
Given a credible leader sequence $({0}, b_1, \ldots, b_n)$, each {credible} leader block $b_h$ defines an epoch. Added to the ledger are the blocks which are pointed to by $b_h$, as well as other blocks reachable from $b_h$ but have not been included in previous epochs.
The list of blocks are sorted topologically, with ties broken by their contents. Since the blocks referenced are mined independently, there can be double spends or redundant transactions. A {transaction} ledger {is created}  by keeping only the first transaction among double spends or redundant transactions.

{Based on the preceding definitions, we have the following properties.}
\begin{lemma}\label{lemma: cj honest block append to honest blockchain}
For $j\in \{0,\ldots, m\}$, if $j$-block $k$ is honest, then both $j$-blockchain $\fj{k}$ and $j$-blockchain $k$ must be {$T^j_k$}-credible.
\end{lemma}
\begin{proof}
For $j= 0,\ldots, m$, the lemma admits essentially the same proof as that for Lemma \ref{lemma: honest block append to honest blockchain}.
\end{proof}

\begin{definition}\label{def: j lagger}
($j$-lagger and $j$-loner) For $j\in\{0,\ldots,m\}$, an honest $j$-block {$k$} is called a {\em $j$-lagger} if it is the only honest $j$-block mined during {$[{T^j_k-\Delta, T^j_k}]$}.
The $j$-lagger is also called a {\em $j$-loner} if it is also the only honest $j$-block mined during $[{T^j_k, T^j_k}+\Delta]$.
\end{definition}

Suppose $0\le s < t$ and $j\in \{0,\ldots, m\}$. Let $\Nj{s}{t}= N^j_{t}-N^j_s$ denote the total number of honest $j$-blocks mined during time interval $(s,t]$. Let $\Xj{s}{t}$ denote the total number of $j$-laggers mined during  $(s, t]$. Let $\Yj{s}{t}$ denote the total number of $j$-loners mined during $(s, t]$. Let $\Zj{s}{t}$ denote the total number of adversarial $j$-blocks mined during $(s, t]$. By convention, $\Nj{s}{t} = \Xj{s}{t} = \Yj{s}{t} = \Zj{s}{t} = 0$ for all $s \ge t$.

{\subsection{Analysis of the Prism backbone protocol}}

\begin{definition}
For all non-negative real numbers $0\le s < t$, $0<\delta<\frac{1}{2}$, and integer $0\le j \le m$, define
\begin{align}
    \Ej{s}{t} = \Eaj{s}{t} \cap \Ebj{s}{t} \cap \Ecj{s}{t} \cap \Edj{s}{t}
\end{align}
where
\begin{align}
    \Eaj{s}{t} & =  \left\{(1-\delta)(t-s)\alpha < \Nj{s}{t} < (1+\delta)(t-s)\alpha \right\} \label{cj equ: E1} \\
    \Ebj{s}{t} & =  \left\{(1-\delta)(t-s)g\alpha < \Xj{s}{t} \right\}\label{cj equ: E2} \\
    \Ecj{s}{t} &  = \left\{(1-\delta)(t-s)g^2\alpha < \Yj{s}{t} \right\}\label{cj equ: E3} \\
    \Edj{s}{t} & = \left\{\Zj{s}{t} <(t-s)\beta + (t-s)g^2\alpha\delta\right\}. \label{cj equ: E4}
\end{align}
\end{definition}

\begin{lemma}\label{lemma: cj prob of good events}
For all $0<\delta<\frac{1}{2}$, $0\le s < t$, and $j\in \{0,\ldots,m\}$, we have
\begin{align}
    P\left(\Ej{s}{t}\right)  >1 - 9e^{-\frac{1}{24}\eta(t-s)}.
\end{align}
\end{lemma}
\begin{proof}
For $j= 0,\ldots, m$, the lemma admits essentially the same proof as that for Lemma \ref{lemma: c prob of good events}.
\end{proof}

\begin{definition}
For $j\in \{0,\ldots,m\}$, define
\begin{align}
    \Gj{s}{t} = \bigcap_{a\in [0,s], b \in [t,\infty)} \Ej{a}{b}. \label{equ: cj def typical}
\end{align}
\end{definition}

\begin{lemma}\label{lemma: cj prob of typical event}
For all real numbers $0\le s < t-\coeff$ and $j\in \{0,\ldots,m\}$,
\begin{align}
    P\left(\Gj{s}{t}\right) > 1-  \mu  e^{-\frac{1}{{27}}\eta(t-s)}.
\end{align}
\end{lemma}
\begin{proof}
For $j= 0,\ldots, m$, the lemma admits essentially the same proof as that for Lemma \ref{lemma: c prob of typical event}.
\end{proof}



\begin{lemma} \label{lemma: cj lagger diff height}
For $j\in\{0,\ldots, m\}$, all $j$-laggers have different heights.
\end{lemma}
\begin{proof}
For $j= 0,\ldots, m$, the lemma admits essentially the same proof as that for Lemma \ref{lemma: c lagger diff height}.
\end{proof}

\begin{lemma} \label{lemma: cj unique kth block}
For $j\in\{0,\ldots, m\}$, a $j$-loner is the only honest $j$-block at its height.
\end{lemma}
\begin{proof}
For $j= 0,\ldots, m$, the lemma admits essentially the same proof as that for Lemma \ref{lemma: c unique kth block}.
\end{proof}

\begin{lemma} \label{lemma: cj height concensus}
For $j\in \{0,\ldots, m\}$, if a $t$-credible $j$-blockchain has height $h$, then the heights of all $(t+\Delta)$-credible $j$-blockchains are at least $h$.
\end{lemma}
\begin{proof}
Lemma \ref{lemma: cj height concensus} is obvious by the Definition \ref{def: cj credible voter blockchain}.
\end{proof}

\begin{lemma} \label{lemma: cj growth}
Suppose integer $k$, $j$ and real number $t$ satisfy $t\ge \frac{k}{2\alpha}$ and $j \in \{0,\ldots, m\}$. Under event $\Ej{t-\frac{k}{2\alpha}}{t}$, the  $(k-1)$-deep  prefix of every $j$-blockchain mined by time $t$ must be mined no later than time $t-\frac{k}{2\alpha}$.
\end{lemma}
\begin{proof}
For $j= 0,\ldots, m$, the lemma admits essentially the same proof as that for Lemma \ref{lemma: c growth}.
\end{proof}

\begin{theorem} \label{thm: cj growth}
{(Prism blockchain growth theorem)}
Suppose real numbers $s$, $t$ and integer $j$ satisfy $0\le s < t-\coeft$ and $j\in\{0, \ldots, m\}$. Under event $\Ej{s+\Delta}{t-\Delta}$,
the height of every $t$-credible $j$-blockchain is at least $\growth g\alpha (t-s)$ larger than the maximum height of all $s$-credible $j$-blockchains. As a consequence, the probability that some $t$-credible $j$-blockchain is less than $\growth g\alpha(t-s)$ higher than some $s$-credible $j$-blockchain does not exceed $9e^{-\frac{1}{24}\eta(t-s)}$.
\end{theorem}
\begin{proof}
For $j= 0,\ldots, m$, the theorem admits essentially the same proof as that for Theorem \ref{thm: c growth}.
\end{proof}

\begin{theorem}\label{thm: cj quality}
{(Prism blockchain quality theorem)}
Suppose positive integer $k$, $j$ and real number $t$ satisfy $k \ge \coefk$, $j\in \{0,\ldots, m\}$ and $t \ge \frac{k}{\growth g\alpha}$.
Under event $\Gj{t-\frac{k}{2\alpha}+\Delta}{t-\Delta}$, at least $\left(1-g\right)k$ of the last $k$ blocks of every $t$-credible $j$-blockchain are honest. As a consequence, the probability that more than  $gk$ of the last $k$ blocks of some $t$-credible $j$-blockchain are adversarial does not exceed $\mu {e^{-\frac{\eta}{{27}}\left(\frac{k}{2\alpha}-2\Delta\right)}}$.
\end{theorem}
\begin{proof}
For $j= 0,\ldots, m$, the theorem admits essentially the same proof as that for Theorem \ref{thm: c quality}.
\end{proof}

\begin{definition}\label{def: cj permanent}
For $j\in \{1,\ldots, m\}$, a $j$-block or a sequence (of $j$-blocks) is said to be {\em permanent after time $t$} if the $j$-block or sequence remains in all $r$-credible $j$-blockchains with $r\ge t$.
\end{definition}

\begin{theorem} \label{thm: cj common prefix}
{(Prism common prefix theorem)}
Suppose positive integer $k$, $j$ and real number $t$ satisfy $k \ge \frac{160\alpha(1+\Delta)}{\delta}$ and $t\ge \frac{k}{\growth g\alpha}$ and $j\in \{1,\ldots, m\}$. Under event $\Gj{t-\frac{k}{2\alpha}+\Delta}{t-\Delta}$, the $(k-1)$-deep prefix of a $t$-credible $j$-blockchain must be extended by all {$r$-credible $j$-blockchains with $r\ge t$}.
As a consequence, with probability at least  $1-\mu {e^{-\frac{\eta}{{27}}\left(\frac{k}{2\alpha}-2\Delta\right)}}$, the $(k-1)$-deep prefix is permanent after time $t$.
\end{theorem}
\begin{proof}
For $j= {1},\ldots, m$, the theorem admits essentially the same proof as that for Theorem \ref{thm: c common prefix}.
\end{proof}

Next we will investigate the properties of {credible} leader sequences, which is the bases of {transaction} ledger generation.

\begin{lemma} \label{lemma: cp height concensus}
{If a $t$-credible leader sequence has height $h$, then the heights of all $(t+\Delta)$-credible leader sequences are at least $h$.}
\end{lemma}
\begin{proof}
If a $t$-credible leader sequence has height $h$, there must exist at least one published proposer block with height $h$.
Then Lemma \ref{lemma: cp height concensus} is obvious by the Definition \ref{def: cp credible leader sequence}.
\end{proof}

\begin{lemma} \label{lemma: cp growth}
Suppose integer $k$ and real number $t$ satisfy $t\ge \frac{k}{2\alpha}$. Suppose a $t$-credible leader sequence has height $n$. Under event $\Ezero{t-\frac{k}{2\alpha}}{t}$, a propose block  whose height is less or equal to $n-k+1$ must be mined no later than time $t-\frac{k}{2\alpha}$.
\end{lemma}
\begin{proof}
Under $\Ezero{t-\frac{k}{2\alpha}}{t}$,
the total number of proposer blocks mined during $(t-\frac{k}{2\alpha}, t]$ is upper bounded by
\begin{align}
    \Nzero{t-\frac{k}{2\alpha}}{t}+\Zzero{t-\frac{k}{2\alpha}}{t}
    & < (1+\delta)\alpha \frac{k}{2\alpha} + \beta \frac{k}{2\alpha} + \delta g^2\alpha \frac{k}{2\alpha} \label{cp 3.0}  \\
    & < (1+\delta)\alpha \frac{k}{2\alpha} + \left(1-\frac{81}{40}\delta\right)g^2\alpha \frac{k}{2\alpha} + \delta g^2\alpha \frac{k}{2\alpha} \label{cp 3.005} \\
    & < (1+\delta)\alpha \frac{k}{2\alpha} + \left(1-\frac{41}{40}\delta\right)\alpha \frac{k}{2\alpha} \label{cp 3.01}\\
    & < k \label{cp 3.02}
\end{align}
where \eqref{cp 3.0} is due to \eqref{cj equ: E1} and \eqref{cj equ: E3}, \eqref{cp 3.005} is due to \eqref{equ: alpha beta}, and \eqref{cp 3.01} is due to $g<1$. So the number of proposer block mined during $(t-\frac{k}{2\alpha}, t]$ is at most $k-1$. Therefore, every propose block  whose height is less or equal to $n-k+1$ must be mined no later than time $t-\frac{k}{2\alpha}$.
\end{proof}

\begin{theorem} \label{thm: cp growth}
{(Leader sequence growth theorem)}
Suppose real numbers $s$, $t$ satisfy $0\le s < t-\coeft$. Under event $\Ezero{s+\Delta}{t-\Delta}$,
the height of every $t$-credible leader sequence is at least $\growth g\alpha (t-s)$ larger than the maximum height of all $s$-credible leader sequences {and all $s$-credible $0$-blockchains}. As a consequence, the probability that some $t$-credible leader sequence is less than $\growth g\alpha(t-s)$ higher than some $s$-credible leader sequence {or some $s$-credible $0$-blockchain} does not exceed $9e^{-\frac{1}{24}\eta(t-s)}$.
\end{theorem}
\begin{proof}
Assume {the maximum height of all} $s$-credible leader sequences {is} $\ell$. According to Lemma \ref{lemma: cp height concensus}, the heights of all $(s+\Delta)$-credible leader sequences are at least $\ell$.

Under event $\Ezero{s+\Delta}{t-\Delta}$, during time interval $(s+\Delta, t-\Delta]$ the number of $0$-laggers is lower bounded:
\begin{align}
\Xzero{s+\Delta}{t-\Delta} & > (1-\delta)g\alpha (t-s-2\Delta) \\
& = (1-\delta)g\alpha \frac{t-s-2\Delta}{t-s} (t-s) \\
&> (1-\delta)\left(1-\frac{\delta}{40}\right)g\alpha(t-s) \label{cp 3.03}\\
& > \growth g\alpha(t-s)
\end{align}
where \eqref{cp 3.03} is due to $t-s>\coeft$. According to Lemma \ref{lemma: cj lagger diff height},
these $0$-laggers have different heights, thus there {must exist} a $0$-lagger with height of at least  $\ell + \growth g\alpha(t-s)$ {by} time $t-\Delta$. This height lower bounds the {heights} of all $t$-credible leader sequences.

By Lemma \ref{lemma: cj prob of good events}, the probability that some $t$-credible leader sequence is less than $\growth g\alpha(t-s)$ higher than some $s$-credible leader sequence does not exceed $9e^{-\frac{1}{24}\eta(t-s)}$.
\end{proof}

\begin{lemma} \label{lemma: cp t-credible height}
Suppose positive integer $k$ and real number $t$ satisfy $t \ge \frac{k}{\growth g\alpha}$. Then under event {$\Gzero{t-\frac{k}{2\alpha}+\Delta}{t-\Delta}$}, the height of {every} $t$-credible leader sequence is {at least} $k$.
\end{lemma}
\begin{proof}
{Let $r = \frac{k}{\growth g\alpha}$. Evidently, $r>\frac{k}{2\alpha}$.}
If event
{$\Gzero{t-\frac{k}{2\alpha}+\Delta}{t-\Delta}$} occurs, event  $\Ezero{t-r+\Delta}{t-\Delta}$ occurs.
According to Theorem \ref{thm: cp growth},
every $t$-credible leader sequence is at least $\growth g\alpha r = k$ higher than the maximum height of all $\left(t-r\right)$-credible leader sequences. Hence the proof.
\end{proof}

%

\begin{theorem}\label{thm: cp quality}
{(Leader sequence quality theorem)}
Suppose positive integer $k$ and real number $t$ satisfy $k \ge \coeft$ and $t \ge \frac{k}{\growth g\alpha}$.
Under event $\Gzero{t-\frac{k}{2\alpha}+\Delta}{t-\Delta}$, at least $\left(1-g\right)k$ of the last $k$ blocks of every $t$-credible leader sequence are honest. As a consequence, the probability that more than  $gk$ of the last $k$ blocks of some $t$-credible leader sequence are adversarial does not exceed $\mu {e^{-\frac{\eta}{{27}}\left(\frac{k}{2\alpha}-2\Delta\right)}}$.
\end{theorem}

\begin{draft}
\begin{proof}
Suppose {proposer blocks $(0,b_1,\ldots, b_n)$ is a $t$-credible leader sequence.} According to Lemma \ref{lemma: cp t-credible height}, $n \ge k$.
Let $h$ be the maximum height strictly less than $n-k+1$ such that the earliest  proposer block mined on height $h$ is honest. We have $0\le h \le n-k$. Let $s = R_h$ for convenience. {If $\Gzero{t-\frac{k}{2\alpha}+\Delta}{t-\Delta}$ occurs, $\Ezero{t-\frac{k}{2\alpha}}{t}$ occurs. We have $t-s>  \frac{k}{2\alpha} > \coeft$ by Lemma \ref{lemma: cp growth}.}

Let{
\begin{align}
    y = n-k-h. \label{c 10.9}
\end{align}}
Then $y$ lower bounds the number of adversarial proposer blocks on heights {$\{h+1, \ldots, n-k\}$} because the earliest proposer blocks on these heights are adversarial by definition.
Denote the number of adversarial blocks in proposer blocks $\{b_{n-k+1},\ldots, b_n\}$ as $z$.
Then $z+y$ lower bounds the number of adversarial proposer blocks generated during $(s,t]$. We have
\begin{align}
    z+ y & \le \Zzero{s}{t}  \\
    & < \beta(t-s) + \delta g^2 \alpha(t-s)\label{c 11}
\end{align}
where \eqref{c 11} is due to \eqref{cj equ: E4}.

Denote earliest proposer block on height $h$ as block $b^*$, which is honest by definition. {Then} $0$-blockchain $b^*$ is $s$-credible. According to Theorem \ref{thm: cp growth}, under $\Ezero{s+\Delta}{t-\Delta}$,
the height of {$t$-credible} leader sequence $(0,b_1,\ldots,b_n)$ is at least $(1-\frac{41}{40}\delta)(t-s)g\alpha$ greater than the height of all $s$-credible $0$-blockchains. Since the height difference of $0$-blockchain $b^*$ and $t$-credible leader sequence $(0,b_1,\ldots,b_n)$ is $n-h$, we have
\begin{align}
    n-h > (1-\frac{41}{40}\delta)(t-s)g\alpha.\label{c 12}
\end{align}

Thus, under event $\Gzero{t-\frac{k}{2\alpha}+\Delta}{t-\Delta}$, we have
\begin{align}
    \frac{z}{k} & \le \frac{z+ y}{k+y} \label{c 13} \\
    & = \frac{z+ y}{n-h} \label{c 13.01} \\
    & < \frac{\beta(t-s)+\delta g^2\alpha(t-s)}{\growth  g\alpha (t-s)} \label{c 13.1} \\
    & = \frac{1}{\growth}\left( \frac{\beta}{g\alpha} + \delta g\right) \label{c 13.2} \\
    & < \frac{1}{\growth}\left( \left(1-\frac{81}{40}\delta\right)g + \delta g\right) \label{c 13.3} \\
    & = g, \label{c 13.4}
\end{align}
where \eqref{c 13} is due to $z \le k$, \eqref{c 13.01} is due to {\eqref{c 10.9}}, \eqref{c 13.1} is due to \eqref{c 11} and \eqref{c 12}, and \eqref{c 13.3} is due to \eqref{equ: alpha beta}.

According to Lemma \ref{lemma: cj prob of typical event},
\begin{align}
{P\left(\Gzero{t-\frac{k}{2\alpha}+\Delta}{t-\Delta}\right)} > 1-\mu {e^{-\frac{\eta}{{27}}\left(\frac{k}{2\alpha}-2\Delta\right)}}.
\end{align}

As a consequence, with probability at least $1-\mu {e^{-\frac{\eta}{{27}}\left(\frac{k}{2\alpha}-2\Delta\right)}}$, event {$\Gzero{t-\frac{k}{2\alpha}+\Delta}{t-\Delta}$} occurs, under which the fraction of adversarial blocks in {the last} $k$  blocks of {every} $t$-credible leader sequence is at most $g$. Hence the proof.
\end{proof}
\end{draft}

\begin{definition}
($h$-high prefix) {For every} {$h\in\{{0},\ldots,n\}$}, by the $h$-high prefix of {$t$-credible} leader sequence $(0,\bl{1}, \ldots, \bl{n})$ we mean {the sequence of} proposer blocks $(0, \ldots, b_h)$.
\end{definition}

\begin{definition}\label{def: cp permanent}
A proposer block or a sequence of proposer blocks is said to be {\em permanent after time $t$} if the proposer block or sequence remains in all $r$-credible leader sequences with $r\ge t$.
\end{definition}

\begin{lemma} \label{lemma: cj mined after}
{Suppose real numbers $s$, $t$ and integer $j$ satisfy $0\le s < t-\coeft$ and $j\in\{0, \ldots, m\}$. Suppose $j$-blockchain $b$ is $t$-credible. Let $\ell = \left \lceil \growth g\alpha(t-s)\right \rceil$. Under event $\Gj{s+\Delta}{t-\Delta}$, an honest $j$-block whose height is greater or equal to $\hj{b} - \ell + 1$ must be mined after time $s$.}
\end{lemma}
\begin{proof}
Suppose honest $j$-block $d$ satisfies ${T^j_d} \le s$. Since $j$-blockchain $d$ is ${T^j_d}$-credible, by Theorem \ref{thm: cj growth} the heights of all $t$-credible $j$-blockchains must be at least $\hj{d} + \growth g \alpha(t-T^j_d)$. Since $j$-blockchain $b$ with height $\hj{b}$ is $t$-credible, we have
\begin{align}
   \hj{d} & \le \hj{b} -  \growth g \alpha(t-{T^j_d}) \\
    & \le \hj{b} - \growth g \alpha(t-s) \label{c 9} \\
     & < \hj{b} - \ell + 1 \label{c 10}
\end{align}
where \eqref{c 10} is by definition of $\ell$. Thus, an honest $j$-block with {height greater or equal to} $\hj{b} - \ell + 1$ must be mined after time $s$.
\end{proof}

\begin{theorem} \label{thm: c permanent leader sequence}
{(Leader sequence common prefix theorem)} Suppose positive integer $k$ satisfies $k\ge \coefk$.
Suppose integer $h$ {and real number $t$ satisfy}
\begin{align}
    t\ge R_h + \frac{k}{\growth(1-g)g\alpha}+\Delta. \label{c 10.01}
\end{align}
{Let}
\begin{align}
    G = \bigcap_{j\in \{1,\ldots, m\}} \Gj{t-\frac{k}{2\alpha}+\Delta}{t-\Delta}.
\end{align}
Then under event $G$, all $t$-credible leader sequences share the same $h$-high prefix, and the prefix is permanent after time $t$.
As a consequence, with probability at least  $1-m\mu {e^{-\frac{\eta}{{27}}\left(\frac{k}{2\alpha}-2\Delta\right)}}$, the $h$-high prefix of all $t$-credible leader sequences is  permanent after time $t$.
\end{theorem}
\begin{short}
\begin{proof}
Theorem \ref{thm: c permanent leader sequence} can be proved based on Theorem \ref{thm: cj growth}, Corollary \ref{coro: cj common prefix}, and Theorem \ref{thm: cj quality}.
\end{proof}
\end{short}
\begin{draft}
\begin{proof}
For convenience, let $s = R_h+\Delta$ {and $\ell = \left \lceil \frac{k}{1-g} \right \rceil$.}
Consider a $t$-credible voter $j$-blockchain $b$. If event $\Gj{t-\frac{k}{2\alpha}+\Delta}{t-\Delta}$ occurs, event $\Ej{s+\Delta}{t-\Delta}$ occurs. By Theorem \ref{thm: cj growth}, the height of $j$-blockchain $b$ is {higher than an $s$-credible $j$-blockchain by} at least
\begin{align}
\growth g\alpha (t-s) > \frac{k}{1-g}, \label{c 10.02}
\end{align}
{where \eqref{c 10.02} is due to \eqref{c 10.01}.} Then $\hj{b} \ge \ell$ because the height is an integer.

Obviously $\ell > k$. If $\Gj{t-\frac{k}{2\alpha}+\Delta}{t-\Delta}$ occurs,  $\Gj{t-\frac{\ell}{2\alpha}+\Delta}{t-\Delta}$ occurs. According to Theorem \ref{thm: cj quality}, in the last $\ell$ blocks of $j$-blockchain $b$, the number of honest ones is at least
\begin{align}
    (1-g)\ell \ge  k.
\end{align}
Thus, the lowest of these honest $j$-blocks, denoted as $j$-block $d$, must be on the $(k-1)$-deep prefix of $j$-blockchain $d$. That is to say, $\hj{b}-\ell+1 \le \hj{d} \le \hj{b}-k+1$.

By Lemma \ref{lemma: cj mined after}, ${T^j_{d}} > s$.
Since $s = R_h+\Delta$, by the voting rule $j$-blockchain $d$ must have voted on all heights less or equal to height $h$. If event $G$ occurs, event $\Gj{t-\frac{k}{2\alpha}+\Delta}{t-\Delta}$ occurs. By Theorem \ref{thm: cj common prefix}, $j$-blockchain $d$ and its votes must be permanent after time $t$.

Such claims can be said for all voter blockchains. That is to say, under event $G$, for all $j\in \{1,\ldots, m\}$ there exists a permanent honest $j$-blockchain which has voted on all height less or equal to $h$. Thus, the $h$-high prefix of all $t$-credible leader sequences is permanent after time $t$.

According to Lemma \ref{lemma: cj prob of typical event},
\begin{align}
    P(G^c) & = P\left(\bigcup_{j\in \{1,\ldots, m\}} \left(\Gj{t-\frac{k}{2\alpha}+\Delta}{t-\Delta}\right)^c \right)  \\
    & < mP\left(\Gj{t-\frac{k}{2\alpha}+\Delta}{t-\Delta}\right) \\
    & < m\mu e^{-\frac{\eta}{{27}}\left(\frac{k}{2\alpha}-2\Delta\right)}.
\end{align}
As a consequence, with probability at least  $1-m\mu {e^{-\frac{\eta}{{27}}\left(\frac{k}{2\alpha}-2\Delta\right)}}$, the $h$-high prefix of $t$-credible leader sequences is  permanent after time $t$.
\end{proof}
\end{draft}

\begin{corollary}
Fix positive integer $h$. {Suppose
\begin{align}
  \epsilon < m\mu e^{-3(1+\Delta)\delta g^2 {\alpha}} \label{c 12.01}
\end{align}
and}
\begin{align}
     t = R_h +   \frac{1}{\growth (1-g) g\alpha}\left({\frac{{54}\alpha}{\eta}} \log\left( \frac{m\mu}{\epsilon} \right) + 4\alpha\Delta+1\right) + \Delta. \label{c 12.02}
\end{align}
Then with probability at least $1-\epsilon$, the $h$-high prefix of {a} $t$-credible leader sequence is  permanent after time $t$.
\end{corollary}
\begin{proof}
Let
\begin{align}
k = \left \lceil \frac{{54}\alpha}{\eta}\log \left( \frac{m\mu}{\epsilon} \right) + 4\alpha \Delta \right \rceil,\label{c 12.03}
\end{align}
{which clearly satisfies $t > R_h + \frac{k}{\growth(1-g)g\alpha}+\Delta$ where $t$ is given in \eqref{c 12.01}.}

{
Note that
\begin{align}
m\mu e^{-\frac{\eta}{{27}}\left(\frac{k}{2\alpha}-2\Delta\right)} & \le
{m\mu e^{-\frac{\eta}{{27}}\left(\frac{{27}}{\eta}\log \left( \frac{m\mu}{\epsilon} \right) + 2 \Delta  - 2\Delta\right)}} \label{c 12.04}\\
& = \epsilon
\end{align}
where \eqref{c 12.04} is due to \eqref{c 12.03}. Also,
\begin{align}
    k & > \frac{{54}\alpha}{\eta}\log \left( \frac{m\mu}{\epsilon} \right) \\
    & > \frac{{54}\alpha}{\eta}\log\left( e^{3(1+\Delta)\delta g^2 {\alpha}} \right) \label{c 12.05} \\
    & = \frac{{54}\alpha}{\delta^2 g^2 \alpha} 3(1+\Delta)\delta g^2 {\alpha} \label{c 12.06} \\
    & > \coefk
\end{align}
where \eqref{c 12.05} is due to \eqref{c 12.01} and \eqref{c 12.06} is due to {\eqref{def: c eta}}.
}

Applying Theorem \ref{thm: c permanent leader sequence} with $k$ given by \eqref{c 12.03}, with probability at least $1- m\mu e^{-\frac{\eta}{{27}}\left(\frac{k}{2\alpha}-2\Delta\right)} > 1-\epsilon$, the $h$-high prefix of a $t$-credible leader sequence is permanent.
\end{proof}


After establishing theorems to ensure the liveness and consistency of Prism blocks, we will discuss the property of Prism transactions. A transaction in Prism blockchains functions the same as that in bitcoin blockchains, which is broadcast and collected by miners to form blocks. A transaction spends from previous transaction outputs as its inputs, and dedicates an amount less or equal to the total input value to new outputs. A transaction may contain multiple inputs (outputs).
See \cite[Chapter 3]{bonneau2016bitcoin} for more detailed introductions to transactions.

\begin{definition}
{(Credible transaction until time $t$) By saying two transactions {have a} conflict we mean there exists at least one output  they both spend from.
A transaction is said to be {\em {\credible}until time $t$} if it
is published by time $t$ and does not {have a} conflict with any other transactions published by time $t$.}  If there is no need to specify $t$ explicitly, we can  simply say a transaction is credible.
\end{definition}


\begin{draft}
\begin{lemma} \label{lemma: c honest leader block include tx}
Suppose honest proposer block $b$ is {in} a $t$-credible leader sequence.
Suppose a  transaction $tx$ {which is \credible until time $t$}  enters a block and the block is broadcast before time $T^0_{b}-\Delta$. Then the {transaction} ledger generated by this $t$-credible leader sequence must include transaction $tx$.
\end{lemma}

\begin{proof}
Suppose the transaction $tx$ enters block $d$ which is {published} by time $T^0_b-\Delta$.
Note that block $d$ may be honest or adversarial, a voter block or a proposer block, and it can be on the credible blockchains or an orphan block.

Denote ${B}$ to be the set of blocks observable {from} proposer block $b$ {(Definition \ref{def: c observable} defines ``observable'')}. Since block $d$ is published by time $T^0_b-\Delta$, we know block $d \in {B}$. Denote  ${C}$ to be a subset of ${B}$ which contains all blocks that can reach block $d$ (Definition \ref{def: c reachable} defines ``reach''). If ${C}$ is empty, block $d$ is not reachable by any blocks {observable} from block $b$. Then proposer block $b$ must reference block $d$ {according to} the Prism protocol. Otherwise, we note that the number of blocks in ${C}$ is finite, and that reference links cannot form a circle. Then there must exist at least one  block which is not referenced by any other block in ${C}$, which must be referenced by proposer block $b$ {according to} the Prism protocol. In both cases, block $d$ is reachable from proposer block $b$.

{As a consequence,} block $d$ is included in the {transaction} ledger by either proposer block $b$ or by previous credible leader blocks in the same $t$-credible leader sequence. Since transaction $tx$ does not conflict with other transactions {by time $t$}, it will not be discarded.
\end{proof}
\end{draft}

\begin{definition}
A transaction is said to be permanent after time $t$ if it remains on the {transaction} ledger of all credible leader sequences after time $t$.
\end{definition}

\begin{lemma} \label{lemma: cp mined after}
{Suppose real numbers $s$ and $t$ satisfy $0\le s < t-\coeft$. Suppose there exists a $t$-credible leader sequence $(0,b_1,\ldots, b_n)$. Let $\ell = \left \lceil \growth g\alpha(t-s)\right \rceil$. Under event $\Gzero{s+\Delta}{t-\Delta}$, an honest proposer block whose height is greater or equal to $n - \ell + 1$ must be mined after time $s$.}
\end{lemma}
\begin{proof}
Suppose honest proposer block $d$ satisfies ${T^0_d} \le s$. {If event $\Gzero{s+\Delta}{t-\Delta}$ occurs, event $\Ezero{T^0_d+\Delta}{t-\Delta}$ occurs}. Since $0$-blockchain $d$ is ${T^0_d}$-credible, by Theorem \ref{thm: cp growth} the heights of all $t$-credible leader sequences must be at least $h^0(d) + \growth g \alpha(t-T^0_d)$. Then we have
\begin{align}
   h^0(d) & \le n -  \growth g \alpha(t-{T^0_d}) \\
    & \le n - \growth g \alpha(t-s) \\
     & < n - \ell + 1 \label{c 10.0}
\end{align}
where \eqref{c 10.0} is by definition of $\ell$. Thus, an honest proposer block with {height greater or equal to} $n - \ell + 1$ must be mined after time $s$.
\end{proof}

\begin{theorem}\label{thm: c permanent tx}
Suppose integer $k$ and real number $r$, $t$ satisfy $k\ge \coefk$ and
\begin{align}
t = r + \frac{2(k+1)}{\growth^2g^2(1-g)^2\alpha} + \Delta. \label{c 13.79}
\end{align}
{Let}
\begin{align}
    G = \bigcap_{j\in \{0,\ldots, m\}} \Gj{t-\frac{k}{2\alpha}+\Delta}{t-\Delta}.
\end{align}
Suppose a  transaction {which is \credible until time $t$} enters a block and the block is published by time $r$.
Then under event $G$, the transaction is permanent after time $t$. As a consequence, with probability of at least $1-(m+1)\mu {e^{-\frac{\eta}{{27}}\left(\frac{k}{2\alpha}-2\Delta\right)}}$, the transaction is permanent after time $t$.
\end{theorem}
\begin{proof}
{
The key to the proof is to identify an honest {proposer} block published early enough but after the block of the transaction under question and invoke Theorem \ref{thm: c permanent leader sequence} and Lemma \ref{lemma: c honest leader block include tx}.
}

Consider a $t$-credible leader sequence $(0,b_1,\ldots,b_n)$.  For convenience, let
\begin{align}
    \ell = \left\lceil \frac{2(k+1)}{\growth (1-g)^2g} \right \rceil \label{c 13.80},
\end{align}
\begin{align}
    u = \left\lceil \frac{2(k+1)}{\growth (1-g)g} \right \rceil, \label{c 13.81}
\end{align}
and
\begin{align}
    s = r+\Delta. \label{c 13.816}
\end{align}
By \eqref{c 13.79},
\begin{align}
s   & = t - \frac{2(k+1)}{\growth^2g^2(1-g)^2\alpha} \\
    & < t - \frac{k}{2\alpha}  \label{c 13.815} \\
    & < t-\coeft.
\end{align}
By \eqref{c 13.815}, if event $\Gzero{t-\frac{k}{2\alpha}+\Delta}{t-\Delta}$ occurs,  event ${\Ezero{s+\Delta}{t-\Delta}}$ occurs.
According to Theorem \ref{thm: cp growth}, the height of {credible leader sequence} $(0,b_1,\ldots, b_n)$ is higher than the height of an $s$-credible leader sequence by at least
\begin{align}
(1-\frac{41}{40}\delta)g\alpha(t-s) = \frac{2(k+1)}{\growth (1-g)^2g} \label{c 13.82}
\end{align}
where \eqref{c 13.82} is due to \eqref{c 13.79} and \eqref{c 13.816}. Then $n \ge \ell$ because the height is an integer.

Obviously $\ell > k$. If $\Gzero{t-\frac{k}{2\alpha}+\Delta}{t-\Delta}$ occurs,  $\Gzero{t-\frac{\ell}{2\alpha}+\Delta}{t-\Delta}$ occurs. According to Theorem \ref{thm: cp quality}, the number of honest proposer blocks in the last $\ell$ blocks of the $t$-credible leader sequence
is at least
\begin{align}
   \left \lceil (1-g)\ell \right \rceil \ge u.
\end{align}
 Denote the lowest honest proposer block in {proposer blocks $\{b_{n-\ell+1},\ldots, b_n\}$} as block $b$. Then block $b$ must be at least $u$-deep. That is to say,
\begin{align}
    n-\ell + 1 \le h^0(b)\le n-u+1. \label{c 13.9}
\end{align}

Obviously $u>k$. If $\Gzero{t-\frac{k}{2\alpha}+\Delta}{t-\Delta}$ occurs,  $\Ezero{t-\frac{u}{2\alpha}}{t}$ occurs. We have $T^0_b <t- \frac{u}{2\alpha}$ by Lemma \ref{lemma: cj growth}.
Then by Definition \ref{def: R_h},
\begin{align}
    R_{h^0(b)} &\le {T^0_b} \\
    & < t - \frac{u}{2\alpha} \\
    & \le  t-\frac{k+1}{\growth(1-g)g\alpha} \\
    & < t - \frac{k}{\growth(1-g)g\alpha} -\frac{1}{g\alpha} \\
    & = t - \frac{k}{\growth(1-g)g\alpha} -\frac{e^{\alpha \Delta}}{\alpha} \label{c 13.95}\\
    & < t-\frac{k}{\growth(1-g)g\alpha} -\Delta \label{c 14}
\end{align}
{where \eqref{c 13.95} is due to \eqref{def: g} and \eqref{c 14} is because $e^x > x$ for all positive $x$}.

If event $G$ occurs, event $ \bigcap_{j\in \{1,\ldots, m\}} \Gj{t-\frac{k}{2\alpha}+\Delta}{t-\Delta}$ occurs. According to Theorem \ref{thm: c permanent leader sequence} and \eqref{c 14}, the $h^0(b)$-deep prefix of $t$-credible leader sequences $(0,b_1,\ldots, b_n)$
is permanent after time $t$.

Note that $\ell = \left \lceil \growth g\alpha(t-s)\right \rceil$ due to \eqref{c 13.79}, \eqref{c 13.80}, and \eqref{c 13.816}. Also, $ h^0(b) \ge n-\ell +1$ due to \eqref{c 13.9}. By \eqref{c 13.815}, $\Gzero{s+\Delta}{t-\Delta}$ occurs.
Applying Lemma \ref{lemma: cp mined after} we have $s < {T^0_b}$.
Then $r = s-\Delta < {T^0_b}-\Delta$. By Lemma \ref{lemma: c honest leader block include tx}, the transaction {which is \credible until time $t$} must be included by the {transaction} ledger generated by $t$-credible leader sequence $(0,b_1,\ldots, b_n)$.  {Hence the transaction is permanent.}

According to Lemma \ref{lemma: cj prob of typical event},
\begin{align}
    P(G^c) & = P\left(\bigcup_{j\in \{0,\ldots, m\}} \left(\Gj{t-\frac{k}{2\alpha}+\Delta}{t-\Delta}\right)^c \right)  \\
    & < (m+1)P\left(\Gj{t-\frac{k}{2\alpha}+\Delta}{t-\Delta}\right) \\
    & < (m+1)\mu e^{-\frac{\eta}{{27}}\left(\frac{k}{2\alpha}-2\Delta\right)}.
\end{align}
As a consequence, with probability at least  $1-(m+1)\mu {e^{-\frac{\eta}{{27}}\left(\frac{k}{2\alpha}-2\Delta\right)}}$, the transaction must be is permanent after time $t$.
\end{proof}

\begin{corollary}
Suppose
\begin{align}
  \epsilon < (m+1)\mu e^{-3(1+\Delta)\delta g^2 {\alpha}} \label{c 15.01},
\end{align}
\begin{align}
k = \left \lceil \frac{{54}\alpha}{\eta}\log \left( \frac{(m+1)\mu}{\epsilon} \right) + 4\alpha \Delta \right \rceil, \label{c 15.03}
\end{align}
and
\begin{align}
     t = r + \frac{2(k+1)}{\growth^2 (1-g)^2 g^2\alpha} + \Delta. \label{c 15.2}
\end{align}
Suppose a transaction {which is \credible until time $t$} enters a block and the block is broadcast by time $r$.
Then with probability at least $1-\epsilon$, the transaction is permanent after time $t$.
\end{corollary}
\begin{proof}
Note that
\begin{align}
(m+1)\mu e^{-\frac{\eta}{{27}}\left(\frac{k}{2\alpha}-2\Delta\right)} &
 \le  {(m+1)\mu e^{-\frac{\eta}{{27}}\left(\frac{{27}}{\eta}\log \left( \frac{(m+1)\mu}{\epsilon} \right) + 2\Delta  - 2\Delta\right)}} \label{c 16.4}\\
& = \epsilon
\end{align}
where \eqref{c 16.4} is due to \eqref{c 15.03}. Also,
\begin{align}
    k & > \frac{{54}\alpha}{\eta}\log \left( \frac{(m+1)\mu}{\epsilon} \right) \\
    & > \frac{{54}\alpha}{\eta}\log\left( e^{3(1+\Delta)\delta g^2 {\alpha}} \right) \label{c 16.5} \\
    & = \frac{{54}\alpha}{\delta^2 g^2 \alpha} 3(1+\Delta)\delta g^2 {\alpha}\label{c 16.6} \\
    & > \coefk
\end{align}
where \eqref{c 16.5} is due to \eqref{c 15.01} and \eqref{c 16.6} is due to  {\eqref{def: c eta}}.

Equality \eqref{c 13.79} clearly holds by assumption.
By Theorem \ref{thm: c permanent tx}, with probability at least $1- (m+1)\mu e^{-\frac{\eta}{{27}}\left(\frac{k}{2\alpha}-2\Delta\right)} > 1-\epsilon$, the transaction which enters a block published by time $r$ is permanent after time $t$.
\end{proof}

Theorem \ref{thm: c permanent tx} illustrates that if a transaction is included in some published block {and is credible}, then the transaction is permanent {after sufficient confirmation time}.
Moreover, the confirmation time is proportional to $\log\frac{1}{\epsilon}$ when we want to ensure at most $\epsilon$ probability of failure.

\section{Conclusion}
In this paper, we have analyzed
the bitcoin and the Prism backbone protocols using {a more simplified and rigorous} framework than in previous analyses. In particular,
{we introduced the concept of credible blockchains and construct a well-defined probability space to describe their properties. We avoid relying on the behavior of miners because it is error prone and  often accompanied by unexpected distortion to default probability distributions. We made no assumption on the miners’ strategy except that their aggregate mining rate is upper bounded.
}
We also assume a continuous-time model with no lifespan limitations and allow the block propagation delays to be arbitrary but bounded.
Under 
the new setting,
we established a blockchain growth theorem, 
a blockchain quality theorem, 
and a common prefix theorem 
for the bitcoin backbone protocol.
We have also proved a blockchain growth theorem and a blockchain quality theorem of the credible leader sequence in the Prism protocol.
We have also shown that the credible leader sequence is 
permanent with high probability after sufficient amount of wait time.
As a consequence, every transaction {which is \credible}will eventually enter the {transaction} ledger and become permanent with probability higher than $1-\epsilon$ after a confirmation time proportional to security parameter $\log\frac{1}{\epsilon}$.
This paper provides explicit security bounds for the bitcoin and the Prism backbone transactions, which improves understanding of both protocols and provides practical guidance to public transaction ledger protocol design.

\section*{Acknowledgement}
“We thank Dr.\ Ling Ren for stimulating discussions and for pointing out two mistakes in previous versions of this paper.  We followed Dr.\ Ren \cite[Lemma 6] {Ren2019Analysis} to fix one of those mistakes.”
\bibliographystyle{ieeetr}
\bibliography{ref}

\end{document}